\documentclass[11pt]{article}
\usepackage{amsmath,amsthm,amssymb}
\usepackage{lineno,hyperref}
\usepackage{geometry}
\usepackage{color}
\usepackage{listings}

\newtheorem{theorem}{Theorem}[section]
\newtheorem{lemma}[theorem]{Lemma}

\newtheorem{define}[theorem]{Definition}
\newtheorem{remark}[theorem]{Remark}
\newtheorem{proc}[theorem]{Procedure}

\def\a{{\mathbf{a}}}

\def\R{{\mathbb{R}}}
\def\Q{{\mathbb{Q}}}

\def\N{{\mathbb{N}}}
\def\CH{{\mathbb{C}}}

\def\P{{\mathcal{P}}}
\def\C{{\mathcal{C}}}
\def\V{{\mathcal{V}}}

\def\M{{\mathcal{M}}}
\def\RC{{\mathcal{R}}}

\def\span{{\hbox{\rm{Span}}}}

\def\ord{{\rm{ord}}}
\def\d{{\hbox{\rm{d}}}}

\lstdefinelanguage{Maple}{
   keywords={if, while, do, else, end, for, from, to,then},
   keywordstyle=\color{blue}\bfseries,
   ndkeywords={class, export, boolean, throw, implements, import, this},
   ndkeywordstyle=\color{darkgray}\bfseries,
   identifierstyle=\color{black},
   sensitive=false,
   comment=[l]{//},
   morecomment=[s]{/*}{*/},
   commentstyle=\color{purple}\ttfamily,
   stringstyle=\color{red}\ttfamily,
   morestring=[b]',
   morestring=[b]"
}

\lstdefinelanguage{SOStools}{
   keywords={syms,sosprogram,monomials,sosineq,sossetobj,sossolve,sosgetsol,sospolyvar},
   keywordstyle=\color{blue}\bfseries,
   ndkeywords={syms,sosprogram,monomials,sosineq,sossetobj,sossolve,sosgetsol},
   ndkeywordstyle=\color{blue}\bfseries,
   identifierstyle=\color{black},
   sensitive=false,
   comment=[l]{//},
   morecomment=[s]{/*}{*/},
   commentstyle=\color{purple}\ttfamily,
   stringstyle=\color{red}\ttfamily,
   morestring=[b]',
   morestring=[b]"
}

\modulolinenumbers[5]


\begin{document}


\title{\bf Lower Bound on Derivatives of\\ Costa's Differential Entropy}
%
\author{Laigang Guo, Chun-Ming Yuan, Xiao-Shan Gao\\
KLMM, Academy of Mathematics and Systems Science\\
 Chinese Academy of Sciences, Beijing 100190, China\\
University of Chinese Academy of Sciences, Beijing 100049, China}
\date{}
\maketitle

\begin{abstract}
\noindent
Several conjectures concern the lower bound for the differential entropy $H(X_t)$
of an $n$-dimensional random vector $X_t$ introduced by Costa.
Cheng and Geng conjectured that $H(X_t)$ is completely monotone,
that is, $C_1(m,n): (-1)^{m+1}(\d^m/\d^m t)H(X_t)\ge0$.
McKean conjectured that Gaussian $X_{Gt}$ achieves the minimum of
$(-1)^{m+1}(\d^m/\d^m t)H(X_t)$ under certain conditions, that is,
$C_2(m,n): (-1)^{m+1}(\d^m/\d^m t)H(X_t)\ge(-1)^{m+1}(\d^m/\d^m t)H(X_{Gt})$.
McKean's conjecture was only considered in the univariate case before:
$C_2(1,1)$ and  $C_2(2,1)$ were proved by McKean
and $C_2(i,1),i=3,4,5$ were proved by Zhang-Anantharam-Geng under the log-concave condition.
In this paper, we prove $C_2(1,n)$, $C_2(2,n)$ and observe that McKean's conjecture
might not be true for $n>1$ and $m>2$. We further propose a weaker version
$C_3(m,n): (-1)^{m+1}(\d^m/\d^m t)H(X_t)\ge(-1)^{m+1}\frac{1}{n}(\d^m/\d^m t)H(X_{Gt})$
and prove $C_3(3,2)$, $C_3(3,3)$, $C_3(3,4)$, $C_3(4,2)$ under the log-concave condition.
A systematical procedure to prove $C_l(m,n)$ is proposed based on semidefinite programming
and the results mentioned above are proved using this procedure.

\end{abstract}



\section{Introduction}

Shannon's {\em entropy power inequality (EPI)}
is one of the most important information inequalities~\cite{Shannon1948}, which has many proofs, generalizations, and applications~\cite{Stam1959,Blachman1965,Lieb1978,VerduGuo2006,Rioul2011,Bergmans1974,Zamir1993,Liu2007,Wang2013}.
In particular, Costa presented a stronger version of the EPI in his seminal paper~\cite{Costa1985}.

Let $X$ be an n-dimensional random vector with {\em probability density} $p(x)$.
For $t>0$, define $X_t\triangleq X+Z_t$, where $Z_t\thicksim N_n(0,tI)$ is an independent standard Gaussian random vector with covariance matrix $t\times I$.
The {\em probability density} of $X_t$ is
\begin{equation}
\label{1.7}
p_t(x_t)=\dfrac{1}{(2\pi t)^{n/2}}\int_{\R^n}p(x)\exp\left(-\dfrac{\|x_t-x\|^2}{2t}\right) \d x_t.
\end{equation}
%
{\em Costa's differential entropy} is defined to be the
differential entropy of $X_t$:
\begin{equation}
\label{1.1}
H(X_t)=-\int_{\R^n} p_t(x_t)\log p_t(x_t)\d x_t.
\end{equation}
%
Costa~\cite{Costa1985} proved that
the {\em entropy power} of $X_t$, given by
$N(X_t)=\dfrac{1}{2\pi e}e^{(2/n)H(X_t)}$
is a concave function in $t$.
More precisely,
Costa proved $(\d/\d t)N(X_t)\ge0$ and $(\d^2/\d^2 t)N(X_t)\le0$.
%
%

Due to its importance, several new proofs and generalizations for Costa's EPI were given.
Dembo~\cite{Dembo1989} gave a simple proof for Costa's EPI via the Fisher information inequality.
Villani~\cite{Villani2000} proved Costa's EPI with advanced theories.
Toscani \cite{Toscani2015} proved that  $({\d^3}/{\d^3 t})N(X_t)\ge0$
if $p_t$ is log-concave.
Cheng and Geng   proposed a conjecture~\cite{Cheng2015}:

\textbf{Conjecture 1.}
$H(X_t)$ is {\em completely monotone} in $t$, that is,
\begin{equation}
\label{eq-D}
C_1(m,n): (-1)^{m+1}(\d^m/\d^m t)H(X_t)\ge0.
\end{equation}
Costa's EPI implies $C_1(1,n)$ and $C_1(2,n)$~\cite{Costa1985},
Cheng-Geng proved $C_1(3,1)$ and $C_1(4,1)$~\cite{Cheng2015}.
In \cite{GYG2020}, the multivariate case of Conjecture 1 was considered and $C_1(3,2)$, $C_1(3,3)$,  $C_1(3,4)$ were proved.

Let $X_{G}$ be an $n$-dimensional Gaussian random vector and  $X_{Gt}\triangleq X_{G}+Z_t$
the Gaussian $X_t$.
McKean~\cite{McKean1966} proved that $X_{Gt}$ achieves the minimum of $(\d/\d t)H(X_t)$ and $-(\d^2/\d^2 t)$ $H(X_t)$ subject to Var$(X_t)=\sigma^2+t$,
and conjectured the general case, that is

\textbf{Conjecture 2.} The following inequality holds {subject to Var$(X_t)=\sigma^2+t$},
\begin{equation}\label{eq-C2}
\begin{array}{ll}
C_2(m,n): (-1)^{m+1}(\d^m/\d^m t)H(X_t)\ge(-1)^{m+1}(\d^m/\d^m t)H(X_{Gt})¡£
\end{array}\end{equation}

McKean proved $C_2(1,1)$ and $C_2(2,1)$~\cite{McKean1966}. Zhang-Anantharam-Geng~\cite{Zhang2018} proved $C_2(3,1)$, $C_2(4,1)$ and $C_2(5,1)$ if the probability density function of $X_t$ is log-concave.
The  work~\cite{Zhang2018,McKean1966} were limited to the univariate case.
In this paper, we consider the multivariate case of Conjecture 2
and will prove $C_2(1,n)$ and $C_2(2,n)$, which
give the exact lower bounds for
$(-1)^{m+1}(\d^m/\d^m t)H(X_t)$ for $m=1,2$.
We also notice that in the multivariate case, Conjecture 2 might not be true for $m>2$
even under the log-concave condition, which motivates us to propose the following weaker conjecture.

\textbf{Conjecture 3.} The following inequality holds {subject to Var$(X_t)=\sigma^2+t$},
\begin{equation}\label{eq-C3}
\begin{array}{ll}
C_3(m,n): (-1)^{m+1}(\d^m/\d^m t)H(X_t)\ge(-1)^{m+1}\frac{1}{n}(\d^m/\d^m t)H(X_{Gt}).
\end{array}\end{equation}

The three conjectures give different lower bounds for
the derivatives of $(-1)^{m+1}H(X_t)$.
Also, Conjecture 2 implies Conjecture 3 and
Conjecture 3 implies Conjecture 1, since $H(X_{Gt})\ge0$~\cite{Zhang2018}.
%

In this paper, we propose a systematical and effective procedure to prove $C_l(m,n)$,
which  consists of three main ingredients.
First, a systematic method is proposed to compute constraints
$R_i,i=1,\ldots,N_1$ satisfied by $p_t(x_t)$ and its derivatives.
The condition that $p_t$ is log-concave can also be reduced to a set of constraints  $\mathcal{R}_j,j=1,\ldots,N_2$.
Second, proof for $C_l(m,n)$ is reduced to the following problem
\begin{equation}
\label{eq-prob}
\exists p_i\in\R \hbox{ and } Q_j \hbox{ s.t. } (E -\sum_{i=1}^{N_1} p_i R_i-\sum_{j=1}^{N_2} Q_j \mathcal{R}_j = S)
\end{equation}
where $Q_j$ is a polynomial in $p_t$ and its derivatives such that $Q_j\ge0$
and
$S$ is a sum of squares (SOS).
Third, problem \eqref{eq-prob} can be solved  with the semidefinite programming (SDP)~\cite{Boyd1,Boyd2}.
There exists no guarantee that the procedure will generate a proof,
but when succeeds, it gives an exact and strict proof for $C_s(m,n)$.

Using the procedure proposed in this paper, we first prove $C_2(1,n)$, $C_2(2,n)$. Then we prove $C_3(3,2)$, $C_3(3,3)$, $C_3(3,4)$ and $C_3(4,2)$ under the condition that $p_t$ is log-concave.
$C_2(3,2)$, $C_2(3,3)$, $C_2(3,4)$, and $C_2(4,2)$ cannot be proved with the above procedure even if $p_t$ is log-concave, which motivates us to propose Conjecture 3.

\begin{table}[ht]
\centering
\begin{tabular}{lcccccccc}
\hline
 & $C_2(3,1)$ & $C_3(3,2)$ & $C_3(3,3)$ & $C_3(3,4)$&  $C_3(4,2)$&$C_2(2,n)$\\
Vars & 3  & 14 & 38 &  38&  33& 6 \\
$N_1$ & 6 & 63 &  512& 512 &  417& 8 \\
$N_2$ & 0 & 0 & 6 & 6 & 3 & 0 \\
Time & 0.18 & 0.53 & 9.00 & 9.02 &  4.49 & 0.32 \\
Proof & Yes & Yes & Yes & Yes& Yes & Yes\\
\hline
\end{tabular}
\caption{Data in computing the SOS  with SDP}
\label{tab1}
\end{table}

In Table \ref{tab1}, we give the data for computing the SOS representation \eqref{eq-prob}
using the Matlab software package in Appendix A,
where
Vars is the number of variables,
$N_1$ and $N_2$ are the numbers of constraints in \eqref{eq-prob}.
Time is the running time in seconds collected on a desktop PC with a 3.40GHz CPU
and 16G memory,
and Proof means whether a proof is given.

The procedure is inspired by the work \cite{Costa1985,Villani2000,Zhang2018,Cheng2015},
and uses basic ideas introduced therein.
In particular, our approach can be basically considered as a generalization
of~\cite{Zhang2018} from the univariate case
to the multivariate case and as a generalization of \cite{GYG2020}
by adding the log-concave constraints.
Also, the log-concave constraints considered in this paper are more general than those
in~\cite{Zhang2018}.

The rest of this paper is organized as follows.
In Section 2, we give the proof procedure and prove $C_2(1,n)$.
In Section 3, we prove $C_2(2,n)$ using the proof procedure.
In Section 4, we prove $C_3(3,2)$, $C_3(3,3)$, and $C_3(3,4)$  under the log-concave condition.
In Section 5, we prove $C_3(4,2)$ under the log-concave condition.
In Section 6, conclusions   are presented.

\section{Proof Procedure}
In this section, we give a general procedure to prove $C_s(m,n)$ for specific values of $l,m,n$.

\subsection{Notations}
  Let $[n]_0 = \{0,1,\ldots,n\}$ and $[n] = \{1,\ldots,n\}$,
 and $x_t=[x_{1,t}, \ldots,x_{n,t}]$.
To simplify the notations, we use $p_t$ to denote $p_t(x_t)$ in the rest of the paper.
Denote
$$\mathcal{P}_n =\{\frac{\partial^h p_t}{\partial^{h_1} x_{1,t}\cdots \partial^{h_n} x_{n,t}}:
h = \sum_{i=1}^n h_i, h_i\in \N\}$$
to be the set of all derivatives of $p_t$ with respect to the differential operators
$\frac{\partial}{\partial x_{i,t}},i=1,\ldots,n$
and $\R[\mathcal{P}_n]$ to be the set of polynomials in $\mathcal{P}_n$ with coefficients in $\R$.
For $v\in\mathcal{P}_n$, let $\ord(v)$   be the order of $v$.
For a monomial $\prod_{i=1}^r v_i^{d_i}$ with $v_i\in {\mathcal{P}}_n$,
its {\em degree}, {\em order}, and {\em total order}
are defined to be $\sum_{i=1}^r d_i$, $\max_{i=1}^r \ord(v_i)$,
 and  $\sum_{i=1}^r d_i\cdot \ord(v_i)$, respectively.
%

A polynomial in $\R[\mathcal{P}_n]$ is called a $k$th-order
{\em differentially homogenous polynomial} or simply
a  $k$th-order {\em differential form},
if all its monomials have degree $k$ and total order $k$.
Let $\M_{k,n}$ be the set of all monomials which have degree $k$ and total order $k$.
Then the set of $k$th-order differential forms is
an $\R$-linear vector space generated by $\M_{k,n}$, which is denoted as $\span_\R(\M_{k,n})$.

We will use Gaussian elimination in $\span_\R(\M_{k,n})$ by treating the monomials as variables.
We always use the {\em lexicographic order for the monomials} to be defined below unless mentioned otherwise.
Consider two distinct derivatives
$v_1=\frac{\partial^k p_t}{\partial^{h_1} x_{1,t}\cdots \partial^{h_n} x_{n,t}}$
and
$v_2=\frac{\partial^k p_t}{\partial^{s_1} x_{1,t}\cdots \partial^{s_n} x_{n,t}}$.
We say $v_1>v_2$ if
$h_l>s_l$ and $h_j=s_j$ for $j=l+1,\ldots,n$.
Consider two distinct monomials $m_1=\prod_{i=1}^{r} v_i^{d_i}$
and $m_2=\prod_{i=1}^{r} v_i^{e_i}$,
where $v_i\in {\mathcal{P}}_n$ and   $v_i< v_j$ for $i < j$.
We define $m_1 > m_2$ if  $d_l > e_l$, and $d_i = e_i$ for $i=l+1,\ldots,r$.
%

From \eqref{1.7}, $p_t:\R^{n+1}\rightarrow \R$ is a function in $x_t$ and $t$.
So each polynomial $f\in\R[\P_n]$ is also a function in $x_t$ and $t$,
$\widetilde{f}(t)=\int_{\R^n}f\d x_t$ is a function in $t$,
and the {\em expectation} of $f$ with respect to $x_t$
$\mathbb{E}[f] \triangleq \int_{\R^n}p_t f\d x_t$  is also a function in $t$.
By $f\ge0$, $\widetilde{f}\ge0$, and $\mathbb{E}[f]\ge0$, we mean $f(x_t,t)\ge0$,  $\widetilde{f}(t)\ge0$, and $\mathbb{E}[f](t)\ge0$ for all $x_t\in\R^n$ and $t>0$.
%

\subsection{The proof procedure}
\label{sec-p2}
In this section, we give the  procedure to prove $C_s(m,n)$,
which consists of four steps.

In step 1, we reduce the proof of $C_s(m,n)$ into the proof of an integral inequality,
as shown by the following lemma whose proof will be given in section \ref{sec-p3}.
\begin{lemma}
\label{lm-pr1}
Proof of $C_s(m,n),s=1,2,3$ can be reduced to show
\begin{equation}
\label{eq-tt1}
\begin{array}{ll}
\displaystyle{\int_{\R^n}\frac{E_{s,m,n}}{p_t^{2m-1}}\d x_t} \ge0
\end{array}
\end{equation}
where $E_{s,m,n} =\sum_{a_1=1}^n\cdots\sum_{a_m=1}^n E_{s,m,n,\a_m}$,
$\a_m=(a_1,\ldots,a_m)$, $E_{s,m,n,\a_m}$ is a $2m$th-order differential form in $\R[\P_{m,n}]$, and
\begin{equation}
\label{eq-pm}
{\P}_{m,n} =\{
\frac{\partial^h p_t}{\partial^{h_1} x_{a_1,t}\cdots \partial^{h_{m}} x_{a_m,t}}:
h\in[2m-1]_0; a_i\in[n],i\in[m]\}.
\end{equation}
\end{lemma}

In step 2, we compute the constraints which are relations satisfied by the probability density $p_t$ of $X_t$.
In this paper, we consider two types of constraints: integral constraints and
log-concave constraints which will be given in Lemmas \ref{lm-CS1} and \ref{lm-CS2}, respectively.
Since $E_{s,m,n}$ in \eqref{eq-tt1} is a $2m$th-order differential form, we need only
the constraints which are   $2m$th-order differential forms.

\begin{define}\label{def-41}
An $m$th-order {\em integral constraint} is a $2m$th-order differential form
$R$ in $\R[\P_{n}]$
such that $\int_{\R^n}\ \frac{R}{p_t^{2m-1}}\d x_t=0$.
\end{define}

\begin{lemma}[\cite{GYG2020}]
\label{lm-CS1}
There is a systematical  method to compute the $m$th-order integral constraints
$\C_{m,n} = \{R_i,i=1,\ldots,N_1\}$.
\end{lemma}

A function $f: \R^n\rightarrow\R$ is called {\em log-concave} if $\log f$
is a concave function.
In this paper, by the {\em log-concave condition}, we mean that  the density function $p_t$   is log-concave.
\begin{define}\label{def-11}
An $m$th-order {\em  log-concave constraint} is a $2m$th-order differential form
$\RC$ in $\R[\P_{n}]$ such that $\RC \ge 0$ under the log-concave condition.
\end{define}

The following lemma computes the log-concave constraints,
whose proof is given in section \ref{sec-p4}.
\begin{lemma}
\label{lm-CS2}
Let  ${\bf H}(p_t)\in\R[\P_n]^{n\times n}$ be the Hessian matrix of $p_t$,
$\nabla p_t =(\frac{\partial p_t}{\partial x_{1,t}},\ldots,\frac{\partial p_t}{\partial x_{n,t}})$,
\begin{equation}\label{eq-L1}
\mathbf{L}(p_t)\triangleq p_t{\bf H}(p_t)-\nabla^T p_t\nabla p_t,
\end{equation}
and $\triangle_{k,l},l=1,\ldots,L_k$  the $k$th-order principle minors of $\mathbf{L}(p_t)$.
Then the $m$th-order log-concave constraints are
\begin{equation}\label{eq-CL}
\CH_{m,n}=\{\prod_{i=1}^s (-1)^{k_i}\triangle_{k_i,l_i} T_{k_1,\ldots,k_s}\,|\,\sum_{i=1}^s k_i \le m\}
\end{equation}
where
$T_{k_1,\ldots,k_s}\in\span_\R(\M_{2m-2\sum_{i=1}^s k_i,n})$ and
$T_{k_1,\ldots,k_s}\ge 0$.
For convenience, denote these constraints as
\begin{equation}\label{eq-CL1}
\CH_{m,n}=\{P_j Q_j,j=1,\ldots,N_2\},
\end{equation}
where $P_j$ represents $\prod_{i=1}^s (-1)^{k_i}\triangle_{k_i,l_i}$
and $Q_j$ is the corresponding $T_{k_1,\ldots,k_s}$.
\end{lemma}

In step 3, we give a procedure to write $E_{s,m,n}$ as an SOS under the constraints, detail of which will be given in section \ref{sec-p5}.
\begin{proc}
\label{lm-sos1}
For $E_{s,m,n}$ in Lemma \ref{lm-pr1},
$\C_{m,n} = \{R_i,i=1,\ldots,N_1\}$ in Lemma \ref{lm-CS1}, and
$\CH_{m,n} = \{P_jQ_j,j=1,\ldots,N_2\}$ in Lemma \ref{lm-CS2},
the procedure computes $e_l\in\R$ and $Q_{j}\in\span_\R(\M_{2(m-\deg P_j),n})$ such that
\begin{eqnarray}
&&E_{s,m,n}-\sum_{i=1}^{N_1} e_i R_i-\sum_{j=1}^{N_2} P_jQ_j =S\quad\hbox{\rm and }\label{eq-S1}\\
&&Q_{j}\ge0,j=1,\ldots,N_2\label{eq-S2}
\end{eqnarray}
where $S$ is an SOS. The procedure is not complete in the sense that it may fail to find
$e_i$ and $Q_j$.
\end{proc}

To summarize the proof procedure, we have
\begin{theorem}
\label{th-m1}
If Procedure \ref{lm-sos1} finds \eqref{eq-S1} and \eqref{eq-S2} for certain $s, m,n$,
then  $C_s(m,n)$ is true.
\end{theorem}
\begin{proof}
By Lemma \ref{lm-pr1}, we have the following proof for   $C_s(m,n)$:
\begin{equation}
\label{eq-proof}
\begin{array}{ll}
\int_\R \frac{{E}_{t,m,n}}{p_t^{2m-1}}\d x_t
&\overset{\eqref{eq-S1}}{=}
\int_\R \frac{\sum_{i=1}^{N_1} e_i R_i+\sum_{j=1}^{N_2} P_jQ_j + S}{p_t^{2m-1}}\d x_t\\[0.2cm]
&\overset{S1}{=}\int_\R \frac{\sum_{j=1}^{N_2} P_jQ_j + S}{p_t^{2m-1}}\d x_t
\overset{S2}{\geq} \int_\R \frac{S}{p_t^{2m-1}}\d x_t
\overset{S3}{\geq}0.
\end{array}
\end{equation}
Equality S1 is true, because $R_i$ is an integral constraint by Lemma \ref{lm-CS1}.
By Lemma \ref{lm-CS2} and \eqref{eq-S2}, $P_jQ_j\ge0$ is true under the log-concave condition,
so inequality S2 is true under the log-concave condition.
If the log-concave condition is not needed, we may   set $Q_j=0$ for all $j$.
Finally, inequality S3 is true, because $S\ge0$ is an SOS.
\end{proof}

\subsection{Proof of Lemma \ref{lm-pr1}}
\label{sec-p3}

Costa~\cite{Costa1985} proved the following basic properties for $p_t$ and $H(X_t)$
\begin{eqnarray}
\frac{\d p_t}{\d t}&=&\frac{1}{2}\nabla^2p_t,\label{H1}\\
\frac{\d H(X_t)}{\d t}&=& -\frac{1}{2}\mathbb{E}[\nabla^2\log p_t] =
\frac{1}{2}\int_{\R^n}\frac{\|\nabla p_t\|^2}{p_t}\d x_t,\label{H2}
%
\end{eqnarray}
where
$\nabla p_t =(\frac{\partial p_t}{\partial x_{1,t}},\ldots,
\frac{\partial p_t}{\partial x_{n,t}})$, $\nabla^2p_t=\sum\limits_{i=1}^{n}\frac{\partial^2p_t}{\partial^2 x_{i,t}}$,
$\mathbb{E}[\nabla^2\log p_t]$
is the {expecttation} of $\nabla^2\log p_t$.
Equation \eqref{H1} shows that $p_t$ satisfies the {\em heat equation}.
%

For $s=1$, Lemma \ref{lm-pr1} was proved in \cite{GYG2020}:
\begin{lemma}[\cite{GYG2020}]
\label{lm-E1}
For $m\in\N_{m>1}$, we have
\begin{equation}
\label{eq-E1}
\begin{array}{ll}
(-1)^{m+1}(\d^m/\d^m t)H(X_t)
=\displaystyle{\int_{\R^n}\frac{E_{1,m,n}}{p_t^{2m-1}(x_t)}\d x_t},
\end{array}
\end{equation}
where $E_{1,m,n}= p_t^{2m-1} [(-1)^{m+1}\frac{1}{2}\frac{\d^{m-1}}{\d^{m-1} t}(\frac{\|\nabla p_t\|^2}{p_t}) =\sum_{a_1=1}^n\cdots\sum_{a_m=1}^n E_{1,m,n,\a_m}$
is a $2m$th-order differential form in $\R[\P_{m,n}]$.
\end{lemma}

To prove  Lemma \ref{lm-pr1} for $s=2,3$,
 we need to compute $(\d^m/\d^m t)H(X_{Gt})$.
Let $X_{G}\thicksim N_n(\mu,\sigma^2I)$ be an $n$-dimensional Gaussian random vector and  $X_{Gt}\triangleq X_{G}+Z_t$, where $Z_t\thicksim N_n(0,tI)$ is introduced in Section 1.
Then $X_{Gt}\thicksim N_n(\mu,(\sigma^2+t)I)$ and the probability density of
$X_{Gt}$ is
$\begin{array}{ll}
  \widehat{p}_t=\frac{1}{(2\pi(\sigma^2+t))^{n/2}}{\rm exp}(-\frac{1}{2(\sigma^2+t)}\|x_t-\mu\|^2).
\end{array}$
%
%
%

\begin{lemma}
\label{lm-T2}
Let $T=\nabla^2{\rm log}p_t$ and $T_G=\nabla^2{\rm log}\widehat{p}_{t}$.
Then  under the log-concave condition, we have
\begin{equation}\label{Ine-chain}
\begin{array}{ll}
\mathbb{E}[(-T)^m]\overset{(a)}{\geq} [\mathbb{E}(-T)]^m
\overset{(b)}{\geq} [\mathbb{E}(-T_G)]^m
\overset{(c)}{=}(-1)^{m+1}\frac{2n^{m-1}}{(m-1)!}(\d^m/\d^m t)H(X_{Gt}).
\end{array}
\end{equation}
 %
\end{lemma}
\begin{proof}
We claim $T\le0$  under the log-concave condition, which
implies inequality $(a)$.
From \eqref{H1},
\begin{equation}\label{eq-T2}
T=\frac{p_t\nabla^2p_t-\|\nabla p_t\|^2}{p_t^2}
=\frac{1}{p_t^2}\sum\limits_{a=1}^{n}(p_t\frac{\partial^2p_t}{\partial^2x_{a,t}}-(\frac{\partial p_t}{\partial x_{a,t}})^2).
\end{equation}
By Lemma \ref{lm-CS2},
under the log-concave condition $\triangle_{1,a}=p_t\frac{\partial^2p_t}{\partial^2x_{a,t}}-(\frac{\partial p_t}{\partial x_{a,t}})^2\leq0$ for $a=1,\ldots,n$, so  $T\leq0$ and the claim is proved.
%

To prove inequality $(b)$, we need the concept of {\em Fisher information} \cite{Rioul2011}:
$J(X_t)\triangleq\mathbb{E}\left(\frac{\|\nabla p_t\|^2}{p_t^2}\right).$
By simple computation, we have
\begin{eqnarray}
&&T_G=\nabla^2{\rm log}\widehat{p}_t=-\frac{n}{\sigma^2+t},\label{ine-c1}\\
&&\mathbb{E}(-T)=-\mathbb{E}(\nabla^2{\rm log}p_t)\overset{\eqref{H2}}{=}\int\frac{\|\nabla p_t(x_t)\|^2}{p_t(x_t)}\d x_t = J(X_t).\label{ine-c2}
\end{eqnarray}
From \cite{park2013,Rioul2011}, we have $J(X_t)\geq J(X_{Gt})$.
Then
$\mathbb{E}(-T)=J(X_t)\geq J(X_{Gt})\overset{\eqref{ine-c2}}{=}\mathbb{E}(-T_G)
\overset{\eqref{ine-c1}}{>}0$, and hence inequality $(b)$.

For equation $(c)$, we first have
$H(X_{Gt})=\frac{n}{2}+\frac{n}{2}{\log(2\pi)}+\frac{n}{2}\log(\sigma^2+t)$
and then equation $(c)$:
$(-1)^{m+1}\frac{\d^m H(X_{Gt})}{\d^m t}=\frac{n(m-1)!}{2(\sigma^2+t)^m}\overset{\eqref{ine-c1}}=\frac{(m-1)!}{2n^{m-1}}[\mathbb{E}(-T_G)]^m.$
\end{proof}

%
\begin{lemma}
\label{lm-E0}
For $m\in\N_{m>1}$, we have
\begin{equation}
\label{eq-E0}
\mathbb{E}[(-T)^m] =
\int_\R^n\frac{E_{0,m,n}}{p_t^{2m-1}}dx_t
\end{equation}
where $E_{0,m,n}=\sum_{a_1=1}^n\cdots\sum_{a_m=1}^n E_{0,m,n,\a_m}$,
$\a_m=(a_1,\ldots,a_m)$, and $E_{0,m,n,\a_m}$ is a $2m$th-order differential form in $\R[\P_{m,n}]$.
\end{lemma}
\begin{proof}
From \eqref{eq-T2}, we have
$\mathbb{E}[(-T)^m]
=\int\frac{(\|\nabla p_t\|^2-p_t\nabla^2p_t)^m}{p_t^{2m-1}}dx_t$,
so $E_{0,m,n} = (\|\nabla p_t\|^2-p_t\nabla^2p_t)^m=\sum_{a_1=1}^n\cdots\sum_{a_m=1}^n E_{0,m,n,\a_m}$,
where $E_{0,m,n,\a_m}$ is a $2m$th-order differentially form in $\R[\P_{m,n}]$,
since $\ord(\|\nabla p_t\|^2-p_t\nabla^2p_t)=2$ and $m>1$.
\end{proof}

We can now prove Lemma \ref{lm-pr1} for $s=2,3$. Let
%
\begin{equation}
\label{eq-e23}
\begin{array}{ll}
E_{2,m,n} = E_{1,m,n} - \frac{(m-1)!}{2n^{m-1}} E_{0,m,n}\\
E_{3,m,n} = E_{1,m,n} - \frac{(m-1)!}{2n^{m}} E_{0,m,n}
\end{array}
\end{equation}
where $E_{1,m,n}$ and $E_{0,m,n}$ are  from Lemmas \ref{lm-E1} and \ref{lm-E0}.
By Lemma  \ref{lm-T2},   $C_s(m,n)$ is true if
$\int_\R^n\frac{E_{s,m,n}}{p_t^{2m-1}}dx_t\ge0$ for $l=2,3$.

As a consequence of Lemma \ref{lm-T2}, we can prove $C_2(1,n)$, that is
\begin{theorem}
Subject to $Var(X_t)=(\sigma^2+t)\times I$,  $(-1)^{n+1}\frac{\d}{\d t}H(X_t)$
achieves the minimum when  $X_t$ is Gaussian with variance $(\sigma^2+t)\times I$
for $t>0$ and $n\geq1$.
\end{theorem}
\begin{proof}
By \eqref{Ine-chain},
 $\mathbb{E}(-T)\geq\mathbb{E}(-T_G)$.
By \eqref{H2} and \eqref{ine-c2},
$(\d/\d t)H(X_t)=\frac{1}{2}\int\frac{\|\nabla p_t(x_t)\|^2}{p_t(x_t)}\d x_t=\frac{1}{2}\mathbb{E}(-T)
\geq\frac{1}{2}\mathbb{E}(-T_G)=(\d/\d t)H(X_{Gt})$. The theorem is proved.
\end{proof}

\subsection{Proof of Lemma \ref{lm-CS2}}
\label{sec-p4}
In this section, we prove Lemma \ref{lm-CS2} which computes the $m$th-order  log-concave constraints.

A symmetric matrix $\mathcal{M}\in\R^{n\times n}$ is called {\em negative semidefinite}
and is denoted as $\mathcal{M}\preceq0$, if all its eigenvalues are nonpositive.
From \cite{Boyd1}, $p_t$ is log-concave if and only if for all $x\in \R^n$ and $t>0$,
$\mathbf{L}(p_t)$ in \eqref{eq-L1} is negative semidefinite.
By the knowledge of linear algebra,  $\mathbf{L}(p_t)\preceq0$ if and only if
\begin{equation}\label{eq-L2}
(-1)^{k}\triangle_{k,l}\ge 0 \hbox{ for } 1\le k\le n, 1\le l \le {{n}\choose{k}}
\end{equation}
where $\triangle_{k,l}$ is a $k$-order principle minors of $\mathbf{L}(p_t)$.
Note that elements of $\mathbf{L}(p_t)$ are quadratic differential forms in $\R[\P_n]$.
Then $(-1)^k\triangle_{k,l}$ is a $k$th-order log-concave constraint.
As a consequence,
$\prod_{i=1}^s (-1)^{k_i}\triangle_{k_i,l_i} Q_{k_1,\ldots,k_s}$ is an $m$th-order log-concave constraint,
if
$Q_{k_1,\ldots,k_s}\in\span_\R(\M_{2m-2\sum_{i=1}^s k_i,n})$ and
$Q_{k_1,\ldots,k_s}\succeq 0$. This proves Lemma \ref{lm-CS2}.

As an illustrative  example, assume that $m=2$, $n=2$. From \eqref{eq-L1},
$$
\mathbf{L}(p_t)=
\begin{bmatrix}
p_t\frac{\partial^2p_t}{\partial^2x_{1,t}}-(\frac{\partial p_t}{\partial x_{1,t}})^2
& p_t\frac{\partial^2p_t}{\partial x_{1,t}\partial x_{2,t}}-\frac{\partial p_t}{\partial x_{1,t}}\frac{\partial p_t}{\partial x_{2,t}}\\
p_t\frac{\partial^2p_t}{\partial x_{1,t}\partial x_{2,t}}-\frac{\partial p_t}{\partial x_{1,t}}\frac{\partial p_t}{\partial x_{2,t}}
& p_t\frac{\partial^2p_t}{\partial^2x_{2,t}}-(\frac{\partial p_t}{\partial x_{2,t}})^2
\end{bmatrix}.
$$
From \eqref{eq-L2},
$\triangle_{1,1}=p_t\frac{\partial^2p_t}{\partial^2x_{1,t}}-(\frac{\partial p_t}{\partial x_{1,t}})^2$, $\triangle_{1,2}=p_t\frac{\partial^2p_t}{\partial^2x_{2,t}}-(\frac{\partial p_t}{\partial x_{2,t}})^2$, $\triangle_{2,1}=|\mathbf{L}(p_t)|$.
From Lemma \ref{lm-CS2}, the second order log-concave constraints  are

$R_{1,1}=-\triangle_{1,1} Q_{1,1}$, where $Q_{1,1} = q_{1,1,1}(\frac{\partial p_t}{\partial x_{1,t}})^2 + q_{1,1,2} (\frac{\partial p_t}{\partial x_{1,t}})(\frac{\partial p_t}{\partial x_{2,t}}) + q_{1,1,3} (\frac{\partial p_t}{\partial x_{2,t}})^2$ and $Q_{1,1} \ge 0$,

$R_{1,2} = -\triangle_{1,2} Q_{1,2}$, where $Q_{1,2} = q_{1,2,1}(\frac{\partial p_t}{\partial x_{1,t}})^2 + q_{1,2,2} (\frac{\partial p_t}{\partial x_{1,t}})(\frac{\partial p_t}{\partial x_{2,t}}) + q_{1,2,3} (\frac{\partial p_t}{\partial x_{2,t}})^2$ and $Q_{1,2} \ge 0$,

$R_{2,1} = \triangle_{2,1}$,
$R_3 = \triangle_{1,1} \triangle_{1,2}$

\noindent where
$Q_{1,1},Q_{1,2}\in\span_\R(\M_{2,2})$ and
$\M_{2,2}=\{
p_t\frac{\partial^2 p_t}{\partial^2 x_{1,t}},
p_t\frac{\partial^2 p_t}{\partial^2 x_{2,t}},
(\frac{\partial p_t}{\partial x_{1,t}})^2,
\frac{\partial p_t}{\partial x_{1,t}}\frac{\partial p_t}{\partial x_{2,t}},
 (\frac{\partial p_t}{\partial x_{2,t}})^2
\}$.
The monomials $p_t\frac{\partial^2 p_t}{\partial^2 x_{2,t}}$ and
$p_t\frac{\partial^2 p_t}{\partial^2 x_{1,t}}$ do not appear in $Q_{1,1}$ and $Q_{1,2}$
due to the condition  $Q_{1,1}\ge0$ and $Q_{1,2}\ge0$.

\subsection{Procedure \ref{lm-sos1}}
\label{sec-p5}
In this section, we present   Procedure \ref{lm-sos1},
which is a modification of the proof procedure given in~\cite{GYG2020}.

\begin{proc}
\label{proc-H}
{\rm Input}: $E_{s,m,n};$
$R_i,i=1,\ldots,N_1$ are $2m$th-order differential forms;
$P_j$ is a $2k_j$th-order differential form for $j=1,\ldots,N_2$.

{\rm Output}: $e_i\in\R$ and $Q_{j}\in\span_\R(\M_{2(m-k_j),n})$ such that
\eqref{eq-S1} and \eqref{eq-S2} are true; or fail meaning that such $e_i$ and $Q_j$
are not found.
\end{proc}
%
%
{\bf S1}.
Treat the monomials in $\M_{m,n}$ as new variables $m_l,l=1,\ldots,N_{m,n}$,
which are all the monomials in $\R[\P_{n}]$ with degree $m$ and total order $m$.
We call $m_lm_s$ a {\em quadratic monomial}.

{\bf S2}.
Write monomials in $\C_{m,n}=\{R_i,i=1,\ldots,N_1\}$ as quadratic monomials if  possible.
Doing Gaussian elimination to $\C_{m,n}$ by treating the monomials
as variables and according to
a monomial order such that a quadratic monomial is less than a non-quadratic monomial, we obtain
$$\widetilde{\C}_{m,n}={ \C}_{m,n,1}\cup { \C}_{m,n,2},$$
where ${ {\C}}_{m,n,1}$ is the set of quadratic forms in $m_i$,
${{\C}}_{m,n,2}$ is the set of   non-quadratic forms, and $\span_\R(\C_{m,n})=\span_\R(\widetilde{\C}_{m,n})$.

{\bf S3}.
There may exist relations among the variables $m_i$, which are called {\em intrinsic
constraints}.
For instance, for $m_1=p_t^2 (\frac{\partial^2 p_t}{\partial^2 x_{1,t}})^2$,
$m_2=p_t (\frac{\partial p_t}{\partial x_{1,t}})^2 \frac{\partial^2 p_t}{\partial^2 x_{1,t}}$,
and $m_3=(\frac{\partial p_t}{\partial x_{1,t}})^4$ in $\M_{4,n}$,
an intrinsic constraint is $m_1m_3-m_2^2=0$.
Add the intrinsic constraints which are quadratic forms in $m_i$ to
${\C}_{m,n,1}$ to obtain
$$\widehat{\C}_{m,n,1}=\{\widehat{R}_i,i=1,\ldots,N_{3}\}.$$

{\bf S4}.
Let $\M_{2(m-k_j),n}=\{m_{j,k},k=1,\ldots,V_j\}$ and
$Q_j = \sum_{k=1}^{V_j} q_{j,k} m_{j,k}$, where $q_{j,k}$ are variables to be found later.
Let $\RC_j$ be obtained from $P_jQ_j$ by writing monomials in $P_jQ_j$ as quadratic monomials and eliminating the non-quadratic monomials with $\C_{m,n,2}$,
such that $\RC_j - P_jQ_j \in\span_\R(\C_{m,n})$
and $\RC_j=\sum_{l=1}^{V_j} q_{j,l} h_{j,l}$, where $h_{j,l}\in\R[m_i]$ is a quadratic
form.
Delete those $\RC_j$ which are not quadratic forms in $m_i$
and still denote these constraints as $\RC_j,j=1,\ldots,N_2$.

{\bf S5}.
Let  $\widehat{E}_{s,m,n}$ be obtained from $E_{s,m,n}$
by eliminating the non-quadratic monomials using ${\C}_{m,n,2}$ such that
$E_{s,m,n}-\widehat{E}_{s,m,n}\in\span_\R(\C_{m,n,2})\subset \span_\R(\C_{m,n})$.

{\bf S6}. Since
 $\widehat{E}_{s,m,n}$,
 $\widehat{R}_i,i=1,\ldots,N_{3}$ and
 $ {\RC}_j,j=1,\ldots,N_2$
are quadratic forms in $m_i$,
we can use the Matlab program given in Appendix A to compute $p_i,q_{j,s}\in\R$ such that
\begin{eqnarray}
&\widehat{E}_{s,m,n}
 -\sum_{i=1}^{N_{3}} p_i \widehat{R}_i
 -\sum_{j=1}^{N_2}\RC_j =S ,\label{eq-tt31}\\
&\quad\quad \RC_j = \sum_{l=1}^{V_j}q_{j,l}h_{j,l}, j=1,\ldots,N_2\nonumber\\
&Q_{j}= \sum_{l=1}^{V_j} q_{j,l} m_{j,l}\ge0,j=1,\ldots,N_2\label{eq-tt32}
\end{eqnarray}
where $S=\sum_{i=1}^{N_{m,n}} c_i (\sum_{j=i}^{N_{m,n}} e_{ij} m_j)^2$ is an SOS,
$c_i,e_{ij}\in\R$ and $c_i\ge0$.
If \eqref{eq-tt31} and \eqref{eq-tt32} cannot be found, return fail.

{\bf S7}.
Since $\widehat{R}_i$, $E_{s,m,n}-\widehat{E}_{s,m,n}$, ${\RC}_j - P_jQ_j$ are all in
$\span_\R(\C_{m,n})$,
equations  \eqref{eq-S1} and \eqref{eq-S2} can be obtained from \eqref{eq-tt31}
and \eqref{eq-tt32}, respectively.

\begin{remark}
Let $R$ be an  intrinsic constraint. Then $R$ becomes zero, when
replacing $m_i$ by its corresponding monomial in $\M_{m,n}$.
So $\span_\R(\widehat{\C}_{m,n,1}) = \span_\R({\C}_{m,n,1}) \subset\span_\R({\C}_{m,n})$ in $\R[\P_n]$, that is, we do not need to include the intrinsic constraints in \eqref{eq-tt31}.
But these intrinsic constraints are needed when using the
Matlab program in Appendix A.
\end{remark}
%
%

\subsection{An illustrative example}
\label{sec-p6}
As an illustrative example, we prove $C_2(3,1)$ under the the log-concave condition using
the proof procedure given in section \ref{sec-p2}.
%
Since $n=1$, denote
$x_t = x_{1,t}, f:=f_0:=p_t,f_n:=\frac{\partial^n p_t}{\partial^n x_{1,t}},\, n\in\N_{>0}$.

In step 1, By Lemma \ref{lm-pr1} and \eqref{eq-tt1}, we have
\begin{equation*}
\frac{\d^3H(X_t)}{\d^3 t}- \frac{2!}{2}\mathbb{E}[\frac{(f_1^2-ff_2)^3}{f^{6}}]
 \overset{\eqref{H2}}{=}\int\left(\frac{1}{2}\frac{\d^2}{\d^2 t}\frac{f_1^2}{f}-\frac{(f_1^2-ff_2)^3}{f^{5}}\right)dx_t
\overset{\eqref{eq-tt1}}{=}\displaystyle{\int\frac{E_{2,3,1}}{f^{5}} \d x_t}
\end{equation*}
where
$
E_{2,3,1}=\frac{1}{4}f^4f_{3}^2-\frac{1}{2}f^3f_{1}f_{3}f_{2}+\frac{1}{4}f^4f_{1}f_{5}-\frac{11}{4}f^2f_{1}^2f_{2}^2
-\frac{1}{8}f^3f_{1}^2f_{4}+f^3f_{2}^3+3ff_{1}^4f_{2}-f_{1}^6
$
is a $6$th-order differential form.

In step 2, we compute the constraints with Lemmas \ref{lm-CS1} and  \ref{lm-CS2}.
With Lemma \ref{lm-CS1}, we find 6 third order constraints~\cite{GYG2020}: $\C_{3,1}=\{R_{i}, i=1,\ldots,6\}$:
{\small
\begin{equation*}
\begin{array}{ll}
{R}_{1} = 5ff_{1}^4f_{2}-4f_{1}^6,
&{R}_{2} = 2f^3f_{1}f_{2}f_{3}+f^3f_{2}^3-2f^2f_{1}^2f_{2}^2,\\
{R}_{3} = f^4f_{1}f_{5}+f^4f_{2}f_{4}-f^3f_{1}^2f_{4},
&{R}_{4} = f^3f_{1}^2f_{4}+2f^3f_{1}f_{2}f_{3}-2f^2f_{1}^3f_{3},\\
{R}_{5} = f^2f_{1}^3f_{3}+3f^2f_{1}^2f_{2}^2-3ff_{1}^4f_{2},
&{R}_{6} = f^4f_{2}f_{4}+f^4f_{3}^2-f^3f_{1}f_{2}f_{3}.
\end{array}\end{equation*}
With Lemma \ref{lm-CS2}, we have one third order order log-concave  constraint: $\CH_{3,1}=\{P_1Q_1\}$, where $P_1=ff_2-f_1^2$,
$Q_1\in\span_\R(\M_{4,1})$, and $Q_1\ge0$.

In step 3, we use Procedure \ref{proc-H} to compute the SOS representation \eqref{eq-S1}
and \eqref{eq-S2} with input $E_{2,3,1},\C_{3,1}=\{R_i,i=1,\ldots,6\},P_1=f_1^2-ff_2$.

{\bf S1}. The new variables are $\M_{3,1}=\{m_{1} = f^2f_{3}, m_{2} = ff_{1}f_{2}, m_{3} = f_{1}^3\}$,
which  are listed from high to low in the lexicographical monomial order.

{\bf S2}. Writing monomials in $\C_{3,1}$ as quadratic monomials in $m_i$ if  possible and
doing Gaussian elimination to $\C_{3,1}$, we have
{\small
\begin{equation*}
\begin{array}{ll}
{{\C}}_{3,1,1}=\{
 \widehat{R}_1=5m_{2}m_{3}-4m_{3}^2,
 & \widehat{R}_2=m_{1}m_{3}+3m_{2}^2-\frac{12}{5}m_{3}^2\},\\
{{\C}}_{3,1,2}=\{
\widetilde{R}_1=f^3f_2^3+2m_1m_2-2m_2^2,~~~~
&\widetilde{R}_2=f^4f_1f_5-m_1^2+3m_1m_2+6m_2^2-\frac{24}{5}m_3^2,\\~~~~
\hskip25pt \widetilde{R}_3=f^4f_2f_4+m_1^2-m_1m_2,~~~~
&\widetilde{R}_4=f^3f_1^2f_4+2m_1m_2+6m_2^2-\frac{24}{5}m_3^2\}.
\end{array}\end{equation*}}

{\bf S3}. There exist no intrinsic constraints
and thus ${\widehat{\C}}_{3,1,1}=\{\widehat{R}_1,\widehat{R}_2\}$ and $N_3=2$.

{\bf S4}. $\M_{4,1}=\{f^3f_4, f^2f_1f_3, f^2f_2^2, ff_1^2f_2, f_1^4\}$.
Then $Q_1=q_{3}f^2f_2^2+q_{1,2}ff_1^2f_2+q_{1,3}f_1^4$.
Monomials $f^3f_4, f^2f_1f_3$ do not appear in $Q_1$ due to $Q_1\ge0$.
%
%
Writing monomials in $P_1Q_1$ as quadratic monomials if possible and using $\C_{3,1,2}$ to eliminate  non-quadratic monomials, we obtain
 $\mathcal{{R}}_1 = P_1Q_1 -(\frac{1}{5}q_{1,2}\widehat{R}_1-q_{1,1}\widetilde{R}_1-\frac{1}{5}q_{1,3}\widehat{R}_1) = q_{1,1}(2m_{1}m_{2}-m_{2}^2)+q_{1,2}(\frac{4}{5}m_{3}^2-m_{2}^2)+\frac{q_{1,3}}{5}m_{3}^2$.

{\bf S5}. Writing  $E_{2,3,1}$ as a quadratic form in $m_i$, we have
$$\begin{array}{ll}
\widehat{E}_{2,3,1}\!\!\!&=E_{2,3,1}-\frac{3}{5}\widehat{R}_{1}-\widetilde{R}_{1}-\frac{1}{4}\widetilde{R}_2+\frac{1}{8}\widetilde{R}_4
=\frac{1}{2}m_{1}^2-3m_{1}m_{2}-\frac{3}{2}m_{2}^2+2m_{3}^2.
\end{array}$$

{\bf S6}. Since  $\widehat{E}_{3,1}$, $\widehat{R}_1$, $\widehat{R}_2$, $ {\mathcal{R}}_1$ are quadratic forms in $m_i$,
we can use the Matlab program in Appendix A to obtain the following SOS representation
\begin{equation}
\label{eq-cc3}
\begin{array}{ll}
\widehat{E}_{2,3,1}=\sum_{i=1}^{2}p_i\widehat{R}_i+{\mathcal{R}}_1+\sum_{i=1}^3 c_i (\sum_{j=i}^3 e_{i,j} m_j)^2,\ \ \ \
P_1\succeq0,
\end{array}
\end{equation}
where $p_1=\frac{6}{5},\ p_2=-2$,
$c_1=\frac{1}{2}$,
$e_{3}=1, e_{1,2}=-3, e_{1,3}=2$,
$q_{1}=q_{2}=q_{3}=c_2=c_3=0$.

{\bf S7}.
%
Since $q_{1}=q_{2}=q_{3}=0$, the log-concave constraint $\RC_1$ is not needed and we obtain
$${E}_{2,3,1}=\frac{3}{4}R_{1}+R_{2}+\frac{1}{4}R_{3}+\frac{1}{8}R_{4}-\frac{7}{4}R_{5}-\frac{1}{4}R_{6}+\sum_{i=1}^3 c_i (\sum_{j=i}^3 e_{i,j} m_j)^2$$
From Theorem \ref{th-m1}, a proof for $C_{2}(3,1)$ is given based on the above SOS representation.

\section{Proof of $C_2(2,n)$}
\label{sec-epi}
In this section, we prove $C_2(2,n)$ using the procedure given in section \ref{sec-p2}, that is,
\begin{theorem}
Subject to $Var(X_t)=(\sigma^2+t)\times I$, Gaussian $X_t$ with variance $(\sigma^2+t)\times I$ achieves the minimum of $(-1)^{n+1}\frac{d^2}{d^2 t}H(X_t)$ for $t>0$ and $n\geq1$.
\end{theorem}
The log-concave conditions are not needed, so we may set $Q_j=0$
and  compute $e_i\in\R$ such that
${E}_{2,2,n}-\sum_{i=1}^{N_1} e_i R_i =S$ in \eqref{eq-S1}.

\subsection{Compute $E_{2,2,n}$}
In step 1, we compute ${E}_{2,2,n}$
with \eqref{eq-e23}:
{\small \begin{eqnarray}\label{2.7}
-\frac{\d^2 H(X_t)}{\d^2t}-\frac{1}{2n}\mathbb{E}(\frac{\|\nabla p_t\|^2-p_t\nabla^2 p_t}{p_t^2})^2
=\int\frac{E_{2,2,n}}{p_t^3}\d x_t
\end{eqnarray}}
where
\begin{equation}\label{E2n}
\begin{array}{ll}
E_{2,2,n}&=
-\frac{\d}{\d t}(\frac{\|\nabla p_t\|^2}{2p_t})-
 \frac{1}{2n}\frac{(\|\nabla p_t\|^2-p_t\nabla p_t)^2}{p_t^3}\\
&=-\frac{1}{2}p_t^2\nabla p_t\cdot \nabla(\nabla^2p_t)+\frac{1}{4}p_t\|\nabla p_t\|^2\nabla^2p_t-\frac{1}{2n}(\|\nabla p_t\|^2-p_t\nabla^2 p_t)^2\\[0.2cm]
&=\sum\limits_{a=1}^n\sum\limits_{b=1}^{n}(T_{1,a,b}-\frac{1}{2n}T_{2,a,b}),\hbox{ and}\\[0.2cm]
T_{1,a,b}&=-\frac{1}{2}p_t^2\frac{\partial p_t}{\partial x_{a,t}}\frac{\partial^3p_t}{\partial x_{a,t}\partial^2x_{b,t}}+\frac{1}{4}p_t(\frac{\partial p_t}{\partial x_{a,t}})^2\frac{\partial^2p_t}{\partial^2x_{b,t}}\\[0.2cm]
T_{2,a,b}&=((\frac{\partial p_t}{\partial x_{a,t}})^2-p_t\frac{\partial^2p_t}{\partial^2x_{a,t}})((\frac{\partial p_t}{\partial x_{b,t}})^2-p_t\frac{\partial^2p_t}{\partial^2x_{b,t}})
\end{array}\end{equation}

\subsection{The second order constraints}
\label{sec-epi1}

In step 2, we compute the second order integral constraints.
Due to the summation structure of $E_{2,2,n}$ in \eqref{E2n},
we introduce the following notations
\begin{equation}
\label{eq-v2}
{\mathcal V}_{a,b} =\{ \frac{\partial^h p_t}{\partial^{h_1} x_{a,t} \partial^{h_2} x_{b,t}}:
h=h_1+h_2\in[3]_0\}
\end{equation}
where $a,b$ are variables taking values in $[n]$.
Then $\P_{2,n}=\cup_{a=1}^n\cup_{b=1}^n \V_{a,b}$.

The second order integral constraints are~\cite{GYG2020}:
\begin{equation}
\label{eq-2cons}
\C_{2,n}=\{R_{i,a,b}^{(2)},R_{j}^{(0)}  \,:\,i=1,\ldots,17; j=1,2; a,b\in[n]\},
\end{equation}
where  $R_{i,a,b}^{(2)}$  can be found in~\cite{GYG2020},
$R^{(0)}_i = \sum_{a=1}^n\sum_{b=1}^n R^{(0)}_{i,a,b},i=1,2$, and
{\small
\begin{equation}\label{2.6}
\begin{array}{ll}
%
%
%
R^{(0)}_{1,a,b}= p_t^2\dfrac{\partial^3p_t}{\partial x_{a,t}\partial^2x_{b,t}}\dfrac{\partial p_t}{\partial x_{a,t}}
+\dfrac{\partial^2p_t}{\partial^2 x_{a,t}}\left[p_t^2\dfrac{\partial^2p_t}{\partial^2x_{b,t}}
-p_t\left(\dfrac{\partial p_t}{\partial x_{b,t}}\right)^2\right],\\[0.3 cm]
R^{(0)}_{2,a,b}=p_t\dfrac{\partial^2p_t}{x_{a,t}^2}\left(\dfrac{\partial p_t}{\partial x_{b,t}}\right)^2
+2\dfrac{\partial p_t}{\partial x_{a,t}}\left[p_t\dfrac{\partial^2p_t}{\partial x_{a,t}\partial x_{b,t}}\dfrac{\partial p_t}{\partial x_{b,t}}
-\dfrac{\partial p_t}{\partial x_{a,t}}\left(\dfrac{\partial p_t}{\partial x_{b,t}}\right)^2\right].
\end{array}
\end{equation}}

\subsection{Prove $C_2(2,n)$}

In step 3, we use Procedure \ref{proc-H} to prove $C_2(2,n)$
with $E_{2,2,n}$ and $\C_{2,n}$ in \eqref{eq-2cons} as input.
It suffices to write
\begin{equation}
\label{eq-sosepi}
E_{2,2,n} - \sum_{R\in\C_{2,n}} c_R R = S\ge0
\end{equation}
where $c_R\in\R$ and $S$ is an SOS.
From \eqref{eq-sosepi}, a proof for $C_2(2,n)$ can be given based on Theorem \ref{th-m1}.
%
%
%
%
Since $C_2(2,1)$ was proved in \cite{McKean1966,Zhang2018}, we will consider $C_2(2,n),\ n\ge2$.
The general case    cannot be proved directly with Procedure \ref{proc-H},
due to the existence of the parameter $n$.
We will reduce the general case to a ``finite" problem which can be solved
with Procedure \ref{proc-H}.

From \eqref{2.7} and \eqref{eq-2cons}, to prove \eqref{eq-sosepi}, it suffices to solve
\noindent{\bf Problem I}.
There exist $c_1,c_2\in\R$ and an SOS $S$ such that
$$
\begin{array}{l}
\widetilde{E}_{2,2,n}=\sum\limits_{a=1}^{n}\sum\limits_{b=1}^{n}(T_{1,a,b}-\frac{1}{2n}T_{2,a,b} + c_1R^{(0)}_{1,a,b} + c_2R^{(0)}_{2,a,b}) =S
\end{array}
$$
under the constraints $R^{(2)}_{i,a,b},i=1,\ldots,17$ given in \eqref{eq-2cons}.

Motivated by   symmetric functions,  for any function $f(a,b)$, we have
\begin{equation}\begin{array}{ll}
\label{2.9a}
\sum\limits_{a,b=1}^{n}f(a,b)=\sum\limits_{1\leq a<b}^{n}\left\{\frac{1}{n-1}[f(a,a)+f(b,b)]+[f(a,b)+f(b,a)]\right\}.
\end{array}\end{equation}
By \eqref{2.9a}, we have
\begin{equation*}\label{2.12}
\begin{array}{ll}
L_{2,n}&=\sum\limits_{a=1}^{n}\sum\limits_{b=1}^{n}(T_{1,a,b}-\frac{1}{2n}T_{2,a,b}+
   c_1R^{(0)}_{1,a,b}+c_2R^{(0)}_{2,a,b})\\
&=\sum\limits_{a<b}\left[\frac{1}{n-1}(T_{1,a,a}+T_{1,b,b}-\frac{1}{2n}(T_{2,a,a}+T_{2,b,b})
    +c_1(R^{(0)}_{1,a,a}+R^{(0)}_{1,b,b})
    +c_2(R^{(0)}_{2,a,a}+R^{(0)}_{2,b,b}))\right.\\
   &\left.
  + T_{1,a,b}+T_{1,b,a}-\frac{1}{2n}(T_{2,a,b}+T_{2,b,a})+c_1(R^{(0)}_{1,a,b}+R^{(0)}_{1,b,a})
   +c_2(R^{(0)}_{2,a,b}+R^{(0)}_{2,b,a})\right]\\
&=\sum\limits_{a<b}\left\{\frac{1}{n-1}(T_{1,a,a}+T_{1,b,b})-\frac{1}{2n(n-1)}(T_{2,a,a}+T_{2,b,b})
    +\frac{1}{n-1}[c_1(R^{(0)}_{1,a,a}+R^{(0)}_{1,b,b})
    +c_2(R^{(0)}_{2,a,a}+R^{(0)}_{2,b,b})]\right.\\
   &\left.
  + T_{1,a,b}+T_{1,b,a}-\frac{1}{2n}(T_{2,a,b}+T_{2,b,a})+c_1(R^{(0)}_{1,a,b}+R^{(0)}_{1,b,a})
   +c_2(R^{(0)}_{2,a,b}+R^{(0)}_{2,b,a})\right\}\\
&=\sum\limits_{a<b}\left\{\frac{1}{n-1}(T_{1,a,a}+T_{1,b,b})-\frac{1}{2}(\frac{1}{n-1}-\frac{1}{n})(T_{2,a,a}+T_{2,b,b})
    +\frac{1}{n-1}[c_1(R^{(0)}_{1,a,a}+R^{(0)}_{1,b,b})
    +c_2(R^{(0)}_{2,a,a}+R^{(0)}_{2,b,b})]\right.\\
   &\left.
  + T_{1,a,b}+T_{1,b,a}-\frac{1}{2n}(T_{2,a,b}+T_{2,b,a})+c_1(R^{(0)}_{1,a,b}+R^{(0)}_{1,b,a})
   +c_2(R^{(0)}_{2,a,b}+R^{(0)}_{2,b,a})\right\}\\
&=\sum\limits_{a<b}\left\{\frac{1}{n-1}[(T_{1,a,a}+T_{1,b,b})-\frac{1}{2}(T_{2,a,a}+T_{2,b,b})+c_1(R^{(0)}_{1,a,a}+R^{(0)}_{1,b,b})
    +c_2(R^{(0)}_{2,a,a}+R^{(0)}_{2,b,b})]\right.\\[0.2cm]
&\left.+\frac{1}{2n}[(T_{2,a,a}+T_{2,b,b})-(T_{2,a,b}+T_{2,b,a})]
  +[(T_{1,a,b}+T_{1,b,a})+c_1(R^{(0)}_{1,a,b}+R^{(0)}_{1,b,a})
   +c_2(R^{(0)}_{2,a,b}+R^{(0)}_{2,b,a})]\right\}\\
&=\sum\limits_{a<b}\left(\frac{1}{n-1}L_{1,a,b}+\frac{1}{2n}L_{2,a,b}+L_{3,a,b}\right),
\end{array}
\end{equation*}
where
$$\begin{array}{ll}
L_{1,a,b}= (T_{1,a,a}+T_{1,b,b})-\frac{1}{2}(T_{2,a,a}+T_{2,b,b})+c_1(R^{(0)}_{1,a,a}+R^{(0)}_{1,b,b})
    +c_2(R^{(0)}_{2,a,a}+R^{(0)}_{2,b,b}),\\[0.15cm]
L_{2,a,b}=(T_{2,a,a}+T_{2,b,b})-(T_{2,a,b}+T_{2,b,a}),\\[0.15cm]
L_{3,a,b}=(T_{1,a,b}+T_{1,b,a})+c_1(R^{(0)}_{1,a,b}+R^{(0)}_{1,b,a})
   +c_2(R^{(0)}_{2,a,b}+R^{(0)}_{2,b,a}).
\end{array}
$$

To prove {\bf Problem I}, it suffices to prove

\noindent{\bf Problem II}. There exist $c_1, c_2\in\R$ and SOSs $S_1,S_2,S_3$ such that $L_{1,a,b}=S_1,L_{2,a,b}=S_2,L_{3,a,b}=S_2$ under the  constraints $R^{(2)}_{i,a,b},i=1,\ldots,{17}$.

In {\bf Problem II}, the subscripts $a$ and $b$ are fixed
and we can  prove {\bf Problem II} with  Procedure \ref{proc-H}
with $L_{1,a,b},L_{2,a,b},L_{3,a,b}$ and $R^{(2)}_{i,a,b},i=1,\ldots,{17}$ as input.

Step {\bf S1}. The new variables are all the monomials in $\R[{\mathcal V}_{a,b}]$
with degree 2 and total order 2 (${\mathcal V}_{a,b}$ is defined in \eqref{eq-v2}):
$$
\begin{array}{ll}
&
m_1=\left(\frac{\partial p_t(x_t)}{x_{a,t}}\right)^2,\
m_2=\left(\frac{\partial p_t(x_t)}{x_{b,t}}\right)^2,\
m_3=\frac{\partial p_t(x_t)}{\partial x_{a,t}}\frac{\partial p_t(x_t)}{x_{b,t}},\\
&
 m_4=p_t(x_t)\frac{\partial^2p_t(x_t)}{\partial x_{a,t}\partial x_{b,t}},\
m_5=p_t(x_t)\frac{\partial^2p_t(x_t)}{\partial^2 x_{a,t}},\ m_6=p_t(x_t)\frac{\partial^2p_t(x_t)}{\partial^2x_{b,t}}.
\end{array}
$$

Step {\bf S2}. We  obtain
${\C}_{2,n,1}=\{\widehat{R}_i,i=1,\ldots,7\}$
and ${\C}_{2,n,2}=\{\widetilde{R}_i,i=1,\ldots,10\}$ using Gaussian elimination, where
$$\begin{array}{ll}
\widehat{R}_1=m_1m_6-2m_3^2+2m_3m_4,\ &\widehat{R}_2=-2m_2m_3+m_2m_4+2m_3m_6,\\
\widehat{R}_3=-2m_2^2+3m_2m_6,\ &\widehat{R}_4=-2m_1m_3+m_1m_4+2m_3m_5,\\
\widehat{R}_5=m_2m_5-2m_3^2+2m_3m_4,\ &\widehat{R}_6=-2m_2m_3+3m_2m_4,\\
\widehat{R}_7=-2m_1^2+3m_1m_5. &\\
\widetilde{R}_1=p_t^2\frac{\partial p_t}{\partial x_{b,t}}\frac{\partial^3p_t}{\partial^3 x_{b,t}}-m_2m_6+m_6^2
& \widetilde{R}_2=p_t^2\frac{\partial p_t}{\partial x_{a,t}}\frac{\partial^3p_t}
{\partial^3 x_{a,t}}-m_1m_5+m_5^2\\
\widetilde{R}_3=p_t^2\frac{\partial p_t}{\partial x_{a,t}}\frac{\partial^3p_t}{\partial x_{a,t}\partial^2x_{b,t}} -m_3m_4+m_4^2
& \widetilde{R}_4=p_t^2\frac{\partial p_t}{\partial x_{b,t}}\frac{\partial^3p_t}{\partial^2 x_{a,t}x_{b,t}}-m_3m_4+m_4^2.
\end{array}$$
$\widetilde{R}_k,k=5,\ldots,10$ are not given, because they are not used in the proof.

Step {\bf S3}. There exists one intrinsic constraint: $\widehat{R}_{8}=m_1m_2-m_3^2$ and $N_{3}=8$.

We do not need Step {\bf S4}, science there exist no log-concave constraints.

Step {\bf S5}. Eliminating the non-quadratic monomials in $L_{1,a,b}$, $L_{2,a,b}$, and $L_{3,a,b}$
using $\C_{2,n,2}$, and doing further reduction by $\C_{2,n,1}$, we have
\begin{equation*}
\begin{array}{ll}
\widehat{L}_{1,a,b}
&={L}_{1,a,b}+(\frac{1}{2}-c_1)\widetilde{R}_1+(\frac{1}{2}-c_1)\widetilde{R}_2-(\frac{1}{4}+c_2)\widehat{R}_3-(\frac{1}{4}+c_2)\widehat{R}_7=0,\\[0.2cm]
\widehat{L}_{2,a,b}
&={L}_{2,a,b}-2\widehat{R}_1+\frac{1}{2}\widehat{R}_3-2\widehat{R}_5+\frac{1}{2}\widehat{R}_7\\[0.2cm]
&=-\frac{1}{2}m_{1}m_{5}-\frac{1}{2}m_{2}m_{6}+6m_{3}^2-8m_{3}m_{4}+m_{5}^2-2m_{5}m_{6}+m_{6}^2,\\[0.2cm]
\widehat{L}_{3,a,b}
&=L_{3,a,b}+(\frac{1}{2}-c_1)\widetilde{R}_3+(\frac{1}{2}-c_1)\widetilde{R}_4+(c_1-c_2-\frac{1}{4})\widehat{R}_1+(c_1-c_2-\frac{1}{4})\widehat{R}_5\\[0.2cm]
&=m_{3}^2-2m_{3}m_{4}+m_{4}^2+c_{1}(-4m_{3}^2+6m_{3}m_{4}-2m_{4}^2+2m_{5}m_{6})
\end{array}\end{equation*}
which are quadratic forms in $m_i$.

Step {\bf S6}.
Using the Matlab program in Appendix A, we obtain the following SOS representation
\begin{equation}\label{2.16}
\begin{array}{ll}
\widehat{L}_{1,a,b}=0,\ \ \widehat{L}_{2,a,b}=\sum\limits_{k=1}^{8}p_k\widehat{R}_{k}+(m_{1}-m_{2}-m_{5}+m_{6})^2,\ \ \widehat{L}_{3,a,b}=(m_{3}-m_{4})^2,
\end{array}
\end{equation}
where
$p_1=\frac{1}{2},\ p_2= \frac{1}{2},\ p_{3}=2,\ p_6=-2,\ p_7=-2,
c_1=c_2=p_4=p_5=p_8=0.$
%
%
So, {\bf Problem II} is solved and thus $C_2(2,n)$ is proved.

\section{Proof of $C_3(3,n)$ for $n=2,3,4$ under the log-concave condition}
\label{sec-3}
We use the procedure in section \ref{sec-p2} to prove
$C_3(3,n)$ for $n=2,3,4$ under the log-concave condition.

\subsection{Compute $E_{3,3,n}$}
\label{section31}
In step 1, we compute $E_{3,3,n}$ in \eqref{eq-tt1} and \eqref{eq-e23}:
\begin{equation}
\label{3.1}
\frac{1}{2}\dfrac{\d^2}{\d  t^2}(\dfrac{\|\nabla p_t\|^2}{p_t})-\frac{1}{n^3}\mathbb{E}(\frac{\|\nabla p_t\|^2-p_t\nabla^2p_t}{p_t^2})^3
 \overset{\eqref{H1}}
=\int_{\R^n}\frac{E_{3,3,n}}{p_t^5}\d x_t
\end{equation}
where $E_{3,3,n}=\sum_{a=1}^{n}\sum_{b=1}^{n}\sum_{c=1}^{n}E_{3,a,b,c}$ and
{\small
\begin{equation*}\begin{array}{ll}
E_{3,a,b,c}&\!\!\!=\frac{p_t^4}{4}\frac{\partial^3p_t}{\partial x_{a,t}\partial^2x_{c,t}}\frac{\partial^3p_t}{\partial x_{a,t}\partial^2x_{b,t}}
-\frac{p_t^3}{4}\frac{\partial p_t}{\partial x_{a,t}}\frac{\partial^3p_t}{\partial x_{a,t}\partial^2x_{b,t}}\frac{\partial^2p_t}{\partial^2x_{c,t}}
+\frac{p_t^4}{4}\frac{\partial p_t}{\partial x_{a,t}}\frac{\partial^5p_t}{\partial x_{a,t}\partial^2x_{b,t}\partial^2x_{c,t}}\\[0.2cm]
&-\frac{p_t^3}{4}\frac{\partial p_t}{\partial x_{a,t}}\frac{\partial^3p_t}{\partial x_{a,t}\partial^2x_{c,t}}\frac{\partial^2p_t}{\partial^2x_{b,t}}
+\frac{p_t^2}{4}\left(\frac{\partial p_t}{\partial x_{a,t}}\right)^2\frac{\partial^2p_t}{\partial^2x_{b,t}}\frac{\partial^2p_t}{\partial^2x_{c,t}}
-\frac{p_t^3}{8}\left(\frac{\partial p_t}{\partial x_{a,t}}\right)^2\frac{\partial^4p_t}{\partial^2x_{b,t}\partial^2x_{c,t}}\\[0.2cm]
&-\frac{1}{n^3}[(\frac{\partial p_t}{\partial x_{a,t}})^2-p_t(\frac{\partial^2p_t}{\partial^2 x_{a,t}})][(\frac{\partial p_t}{\partial x_{b,t}})^2-p_t(\frac{\partial^2p_t}{\partial^2 x_{b,t}})][(\frac{\partial p_t}{\partial x_{c,t}})^2-p_t(\frac{\partial^2p_t}{\partial^2 x_{c,t}})].
\end{array}\end{equation*}}

\subsection{Compute the third order constraints}
\label{sec-31}
In step 2, we obtain the third order constraints.
Similar to \eqref{eq-v2}, we introduce the notation
\begin{equation}\label{eq-v3}
{\mathcal V}_{a,b,c} =\{
\dfrac{\partial^h p_t}{\partial^{h_1} x_{a,t}\partial^{h_2} x_{b,t}\partial^{h_3} x_{c,t}}:
h=h_1+h_2+h_3 \in\{0,1,\cdots,5\}\}
\end{equation}
where $a,b,c$ are variables taking values in $[n]$.
Then $\P_{3,n}=\cup_{a=1}^n\cup_{b=1}^n\cup_{c=1}^n \V_{a,b,c}$.

The third order integral constraints are~\cite{GYG2020}:
\begin{equation}
\label{eq-3consn}
\C_{3,n}=\{R_{i,a,b,c}^{(3)},\,:\,i=1,\ldots,955; a,b,c\in[n]\},
\end{equation}
where $R_{i,a,b,c}^{(3)}$ can be found in \cite{GYG2020}.
Note that we do not use all the third order constraints in \cite{GYG2020}.

From Lemma \ref{lm-CS2}, we can compute the third order log-concave constraints:
\begin{equation}
\label{eq-consl32}
\mathbb{C}_{3,2}=\{\RC_1=-\triangle_{1,1}Q_{1},\RC_2=-\triangle_{1,2}Q_{2},\RC_3=\triangle_{2,1}Q_{3}\},
\end{equation}
where $Q_1,Q_2\in\span_\R(\M_{4,4})$ and $Q_3\in\span_\R(\M_{2,2})$.
Note that $\mathbb{C}_{3,2}$ does not contain all the log-concave constraints in
Lemma \ref{lm-CS2}. The constraints $\mathbb{C}_{3,2}$ are enough for our purpose in this paper.

For $n>2$, we give certain log-concave constraints in a special form, which are needed in the proof procedure in section 4.3.
Let $\nabla_1 p_t =(\frac{\partial p_t}{\partial x_{a,t}},\frac{\partial p_t}{\partial x_{b,t}},\frac{\partial p_t}{\partial x_{c,t}})$,
$\mathbf{L}_1(p_t)\triangleq p_t{\bf H}_1(p_t)-\nabla_1^T p_t\nabla_1 p_t$,
where
\begin{equation*}
{\bf H}_1(p_t)=\left[ \begin{array}{ccc}
\frac{\partial^2p_t}{\partial^2x_{a,t}} & \frac{\partial^2p_t}{\partial x_{a,t}\partial x_{b,t}} & \frac{\partial^2p_t}{\partial x_{a,t}\partial x_{c,t}}\\
\frac{\partial^2p_t}{\partial x_{a,t}\partial x_{b,t}} & \frac{\partial^2p_t}{\partial^2x_{b,t}} & \frac{\partial^2p_t}{\partial x_{b,t}\partial x_{c,t}}\\
\frac{\partial^2p_t}{\partial x_{a,t}\partial x_{c,t}} & \frac{\partial^2p_t}{\partial x_{b,t}\partial x_{c,t}} & \frac{\partial^2p_t}{\partial^2x_{c,t}}
\end{array}
\right ],
\end{equation*}
and $\triangle'_{k,l},l=1,\ldots,L_k$  the $k$th-order principle minors of $\mathbf{L}_1(p_t)$.
Let $\M'_{k}$ be the set of all monomials in ${\mathcal V}_{a,b,c}$ (defined in \eqref{eq-v3}) which have degree $k$ and total order $k$.
We have
\begin{equation}
\label{eq-consl33}
\mathbb{C}_{3,n}=\{-\triangle'_{1,1}Q_{1,1},
 -\triangle'_{1,2}Q_{1,2},-\triangle'_{1,3}Q_{1,3},
\triangle'_{2,1}Q_{2,1},\triangle'_{2,2}Q_{2,2},\triangle'_{2,3}Q_{2,3},-\triangle'_{3,1}Q_{3,1}\}
\end{equation}
where $Q_{1,i}\in\span_\R(\M'_{4})$, $Q_{2,j}\in\span_\R(\M'_{2})$,
and $Q_{3,1}\in\R$.

\subsection{Proof of $C_3(3,2)$}
\label{section3.1}

The proof follows Procedure \ref{proc-H} with  $E_{3,3,2}$ given in \eqref{3.1}
and the constraints in \eqref{eq-3consn} and \eqref{eq-consl32}
as input.

In Step {\bf S1},
 the new variables are $\M_{3,2}$ and are listed  in the lexicographical monomial order:
$$\begin{array}{ll}
m_1=p_t^2\frac{\partial p_t^3}{\partial^3 x_{2,t}},
m_2=p_t^2\frac{\partial^3 p_t}{\partial x_{1,t}\partial^2 x_{2,t}},
m_3=p_t^2\frac{\partial^3 p_t}{\partial^2 x_{1,t}\partial x_{2,t}},
m_4=p_t^2\frac{\partial p_t^3}{\partial^3 x_{1,t}},\\
m_5=p_t\frac{\partial^2 p_t}{\partial^2 x_{2,t}}\frac{\partial p_t}{\partial x_{2,t}},
m_6=p_t\frac{\partial^2 p_t}{\partial^2 x_{2,t}}\frac{\partial p_t}{\partial x_{1,t}},
m_7=p_t\frac{\partial^2 p_t}{\partial x_{1,t}\partial x_{2,t}}\frac{\partial p_t}{\partial x_{2,t}},\\
m_8=p_t\frac{\partial^2 p_t}{\partial x_{1,t}\partial x_{2,t}}\frac{\partial p_t}{\partial x_{1,t}},
m_9=p_t\frac{\partial^2 p_t}{\partial x_{1,t}^2}\frac{\partial p_t}{\partial x_{2,t}},
m_{10}=p_t\frac{\partial^2 p_t}{\partial x_{1,t}^2}\frac{\partial p_t}{\partial x_{1,t}},\\
m_{11}=\left(\frac{\partial p_t}{\partial x_{2,t}}\right)^3,
m_{12}=\left(\frac{\partial p_t}{\partial x_{2,t}}\right)^2\frac{\partial p_t}{\partial x_{1,t}},
m_{13}=\frac{\partial p_t}{\partial x_{2,t}}\left(\frac{\partial p_t}{\partial x_{1,t}}\right)^2,
m_{14}=\left(\frac{\partial p_t}{\partial x_{1,t}}\right)^3.
\end{array}$$

In Step {\bf S2},
the constraints are $\C_{3,2}=\{R^{(3)}_{j,a,b,c}\,:\,j=1,\ldots,955;a,b,c\in[2]\}$.
%
Removing the repeated ones, we have $N_1=135$.
We  obtain $\C_{3,2,1}$ and $\C_{3,2,2}$ which contain
48 and 52 constraints, respectively.

In Step {\bf S3}, there exist 15 intrinsic constraints:
{\small
\begin{equation*}
\begin{array}{ll}
 m_{5}m_{8}=m_{6}m_{7}, m_{5}m_{10}=m_{6}m_{9}, m_{5}m_{12}=m_{6}m_{11}, m_{5}m_{13}=m_{6}m_{12}, m_{5}m_{14}=m_{6}m_{13},\\
  m_{7}m_{10}=m_{8}m_{9},  m_{7}m_{12}=m_{8}m_{11}, m_{7}m_{13}=m_{8}m_{12}, m_{7}m_{14}=m_{8}m_{13}, m_{9}m_{12}=m_{10}m_{11},\\
   m_{9}m_{13}=m_{10}m_{12},
 m_{9}m_{14}=m_{10}m_{13}, m_{11}m_{13}=m_{12}^2, m_{11}m_{14}=m_{12}m_{13}, m_{12}m_{14}=m_{13}^2.
\end{array}
\end{equation*}}
Thus,  ${\widehat{\C}}_{3,2,1}$ contains 63 constraints and $N_3 = 63$.

In Step {\bf S4}, we obtain $\widehat{\mathbb{C}}(3,2)$ which contains 3 quadratic form constraints.

In Step {\bf S5}, eliminating the non-quadratic monomials in $E_{3,3,2}$
using $\C_{3,2,2}$ to obtain a quadratic form in $m_i$ and then simplifying the quadratic form using $\C_{3,2,1}$, we have
{\small
\begin{equation*}
\begin{array}{ll}
\widehat{E}_{3,3,2}\!\!\!\!
&=-\frac{147}{8}m_{13}^2+\frac{31}{40}m_{14}^2-\frac{5}{2}m_{7}m_{10}+\frac{15}{4}m_{8}^2-\frac{25}{8}m_{9}^2
-\frac{31}{16}m_{9}m_{11}+\frac{207}{8}m_{9}m_{13}-\frac{5}{8}m_{10}^2+\frac{1}{2}m_{1}^2\\
&-\frac{5}{4}m_{1}m_{5}+\frac{31}{40}m_{11}^2+\frac{31}{8}m_{12}^2+\frac{1}{2}m_{4}^2-\frac{5}{2}m_{4}m_{6}-\frac{5}{4}m_{4}m_{7}
+\frac{3}{2}m_{3}^2-\frac{15}{4}m_{7}^2-\frac{5}{4}m_{4}m_{10}\\
& -\frac{5}{8}m_{5}^2+\frac{15}{8}m_{6}^2+\frac{3}{2}m_{2}^2-\frac{15}{4}m_{2}m_{6}.
\end{array}
\end{equation*}}

In Step {\bf S6}, using the Matlab program in Appendix A
with $\widehat{E}_{3,3,2}$, $\widehat{\C}_{3,2,1}$ and $\widehat{\mathbb{C}}_{3,2}$ as input,
we find an SOS representation for $\widehat{E}_{3,3,2}$. Thus, $C_3(3,2)$ is proved under the log-concave condition.
The Maple  program to prove $C_3(3,2)$ can be found in
https://github.com/cmyuanmmrc/codeforepi/.

\begin{remark}
We fail to prove $C_2(3,2)$ even under the log-concave condition similar to the above procedure. Specifically, we cannot find an SOS representation for $\widehat{E}_{2,3,2}$ in Step {\bf S6}.
Since the SDP algorithm is not complete for  problem \eqref{eq-tt31}, we cannot
say that an SOS representation does not exist for $\widehat{E}_{2,3,2}$.
The Maple program for $C_2(3,2)$ can be found in
https://github.com/cmyuanmmrc/codeforepi/.
\end{remark}

\subsection{Proof of $C_3(3,3)$ and $C_3(3,4)$}
In this subsection, we want to prove $C_3(3,3),C_3(3,4)$. Motivated by symmetric functions, for any function $f(a,b,c)$, we have
\begin{equation}\begin{array}{ll}
\label{2.9ab}
\sum\limits_{a,b,c=1}^{n}f(a,b,c)=\sum\limits_{1\leq a<b<c}^{n}\{\frac{2}{(n-1)(n-2)}[f(a,a,a)+f(b,b,b)+f(c,c,c)]
+\frac{1}{n-2}[f(a,a,b)+f(a,b,a)\\
\ \ \ \ \ \ \ +f(b,a,a)+f(a,a,c)+f(a,c,a)+f(c,a,a)+f(b,b,a)+f(b,a,b)+f(a,b,b)+f(b,b,c)\\
\ \ \ \ \ \ \ +f(b,c,b)+f(c,b,b)+f(c,c,a)+f(c,a,c)+f(a,c,c)+f(c,c,b)+f(c,b,c)+f(b,c,c)]\\
\ \ \ \ \ \ \ +[f(a,b,c)+f(a,c,b)+f(b,a,c)+f(b,c,a)+f(c,a,b)+f(c,b,a)]\}.
\end{array}\end{equation}
From \eqref{3.1} and \eqref{2.9ab}, we obtain
$$\begin{array}{ll}
E_{3,3,n}=\sum_{a=1}^{n}\sum_{b=1}^{n}\sum_{c=1}^{n}E_{3,a,b,c}=\sum\limits_{1\leq a<b<c}^{n}J_{3,3,n},
\end{array}$$
where
\begin{equation}\label{J3n}
\begin{array}{ll}
J_{3,3,n}&=\frac{2}{(n-1)(n-2)}[E_{3,a,a,a}+E_{3,b,b,b}+E_{3,c,c,c}]
+\frac{1}{n-2}[E_{3,a,a,b}+E_{3,a,b,a}+E_{3,b,a,a}+E_{3,a,a,c}\\
&+E_{3,a,c,a}+E_{3,c,a,a}+E_{3,b,b,a}+E_{3,b,a,b}+E_{3,a,b,b}+E_{3,b,b,c}+E_{3,b,c,b}
+E_{3,c,b,b}+E_{3,c,c,a}\\
&+E_{3,c,a,c}+E_{3,a,c,c}+E_{3,c,c,b}+E_{3,c,b,c}+E_{3,b,c,c}]+[E_{3,a,b,c}+E_{3,a,c,b}\\
&+E_{3,b,a,c}+E_{3,b,c,a}+E_{3,c,a,b}+E_{3,c,b,a}]
\end{array}\end{equation}
Thus, if we prove $J_{3,3,n}\geq0$, then $E_{3,3,n}\geq0$.
It is clear that $J_{3,3,n}$ contains much smaller terms than $E_{3,3,n}$.

In  $J_{3,3,n}$ given in \eqref{J3n}
and the constraints in \eqref{eq-3consn} and \eqref{eq-consl33},
we may consider
$\frac{\partial}{\partial x_{a,t}}$, $\frac{\partial}{\partial x_{b,t}}$,
and $\frac{\partial}{\partial x_{c,t}}$ as the differential operators
without giving concrete values to $a,b,c$.
%

First, we prove of $C_3(3,3)$ using Procedure \ref{proc-H} with $J_{3,3,3}$ given in \eqref{J3n}
and the constraints in \eqref{eq-3consn} and \eqref{eq-consl33} as the input.
%
%

In Step {\bf S1}, the new variables are $\M'_3=\{m_i,i=1,\ldots,38\}$.

In Step {\bf S2}, the constraints are:
$\C_{3,n}=\{R_{i,a,b,c}^{(3)}\,:\,i=1,\ldots,955\}$, $N_1=955$. We obtain   $\C_{3,n,1}$ and $\C_{3,n,2}$, which contain
350 and 328 constraints, respectively.

In Step {\bf S3}, there exist 189 intrinsic constraints.
In total, ${\widehat{\C}}_{3,n,1}$ contains 539 constraints.
Using $\R$-Gaussian elimination in $\span_\R({\widehat{\C}}_{3,n,1})$
shows that 512 of these 539 constraints are linearly independent,  so $N_3 = 512$.

In Step {\bf S4}, we obtain $\widehat{\mathbb{C}}_{3,n}$ from ${\mathbb{C}}_{3,n}$ which contains 6 constraints.

In Step {\bf S5},
eliminating the non-quadratic monomials in $J_{3,3,3}$
using $\C_{3,n,2}$ and then simplify the expression using $\C_{3,n,1}$, we have
{\small
$$\begin{array}{ll}
\widehat{J}_{3,3,3}\!\!\!\!&=\frac{31}{9}m_{23}^2+\frac{29}{18}m_{22}^2+\frac{88}{135}m_{29}^2-\frac{29}{54}m_{21}^2-\frac{178}{9}m_{23}m_{33}
+\frac{88}{9}m_{20}m_{34}+\frac{202}{9}m_{20}m_{32}-\frac{145}{54}m_{20}^2-\frac{29}{54}m_{28}^2\\
&+\frac{176}{9}m_{33}^2+\frac{88}{135}m_{38}^2+\frac{88}{27}m_{36}^2+\frac{88}{27}m_{30}^2-\frac{145}{54}m_{27}^2+\frac{88}{135}m_{35}^2+\frac{29}{9}m_{25}^2
+\frac{3}{2}m_{4}^2-\frac{44}{27}m_{20}m_{29}\\
&-\frac{2}{9}m_{20}m_{26}+\frac{202}{9}m_{27}m_{37}-\frac{44}{27}m_{27}m_{35}-\frac{29}{27}m_{1}m_{11}+\frac{3}{2}m_{6}^2+\frac{29}{9}m_{19}^2+\frac{3}{2}m_{8}^2+\frac{3}{2}m_{9}^2
-\frac{29}{54}m_{11}^2+\frac{29}{18}m_{12}^2\\
&-\frac{29}{3}m_{18}^2+\frac{29}{18}m_{13}^2-\frac{58}{9}m_{6}m_{20}-\frac{29}{27}m_{7}m_{14}-\frac{58}{27}m_{7}m_{12}-\frac{29}{27}m_{7}m_{21}-\frac{58}{9}m_{8}m_{13}\\
&+\frac{58}{9}m_{9}m_{12}
-\frac{29}{9}m_{8}m_{22}-\frac{29}{27}m_{10}m_{17}-\frac{58}{27}m_{10}m_{13}-\frac{29}{27}m_{10}m_{24}-\frac{29}{27}m_{10}m_{28}-\frac{58}{27}m_{10}m_{22}\\
&-\frac{86}{9}m_{26}m_{32}+\frac{202}{9}m_{26}m_{34}
+\frac{88}{27}m_{31}^2-\frac{29}{9}m_{24}^2-\frac{29}{9}m_{14}^2+\frac{29}{9}m_{15}^2+\frac{29}{9}m_{16}^2-\frac{29}{9}m_{17}^2-\frac{145}{54}m_{26}^2\\
&+\frac{3}{2}m_{3}^2-\frac{58}{27}m_{24}m_{28} -\frac{58}{27}m_{14}m_{21}-\frac{58}{27}m_{17}m_{28}-\frac{29}{9}m_{3}m_{13}-\frac{44}{27}m_{26}m_{29}
-\frac{29}{9}m_{2}m_{12}\\
&-16m_{37}^2-16m_{34}^2+3m_{5}^2+\frac{1}{2}m_{7}^2+\frac{1}{2}m_{10}^2-16m_{32}^2+\frac{1}{2}m_{1}^2.
\end{array}$$}

In Step {\bf S6}, using the Matlab program in Appendix A
with $\widehat{J}_{3,3,3}$, $\widehat{\C}_{3,n,1}$ and $\widehat{\mathbb{C}}_{3,n}$ as input, we find an SOS representation for $\widehat{J}_{3,3,3}$.
Thus, $C_3(3,3)$ is proved.
The Maple  program to prove $C_3(3,3)$ can be found in
https://github.com/cmyuanmmrc/codeforepi/.

To prove $C_3(3,4)$, we just need to replace the input from $J_{3,3,3}$ to $J_{3,3,4}$ in the Step {\bf S5} in the above procedure.
In the same way, $C_3(3,4)$ can be proved.
The Maple  program to prove $C_3(3,4)$ can be found in
https://github.com/cmyuanmmrc/codeforepi/.

\section{Proof of $C_3(4,2)$}
\label{sec-42}

We use the procedure in section \ref{sec-p2} to prove
$C_3(4,2)$ under the log-concave condition.

In step 1, we compute $E_{3,4,n}$ in \eqref{eq-tt1} and \eqref{eq-e23}:
\begin{equation}
\label{4.1}
\frac{1}{2}\dfrac{\d^3}{\d  t^3}(\dfrac{\|\nabla p_t\|^2}{p_t})-\frac{3}{n^4}\mathbb{E}(\frac{\|\nabla p_t\|^2-p_t\nabla^2p_t}{p_t^2})^4
 \overset{\eqref{H1}}
=\int_{\R^n}\frac{E_{3,4,n}}{p_t^7}\d x_t
\end{equation}
where $E_{3,4,n}=\sum_{a=1}^{n}\sum_{b=1}^{n}\sum_{c=1}^{n}\sum_{d=1}^{n}E_{4,a,b,c,d}$. For brevity, we omit the concrete expression of $E_{4,a,b,c,d}$.

In step 2,
based on Lemma \ref{lm-CS1}, we obtain 589 fourth order constraints
\begin{equation}\label{42cons}
\begin{array}{ll}
\C_{4,2}=\{R^{(2)}_{i,1,2}\,:\,i=1,\ldots,589\}\subset\R[\P_{2}]\hbox{ and } N_1=589.
\end{array}
\end{equation}
%
By  Lemma \ref{lm-CS2}, we obtain three $4$th-order log-concave constraints: $$\mathbb{C}_{4,2}=\{-\triangle_{1,1}Q_{1,1},-\triangle_{1,2}Q_{1,2},\triangle_{2,1}Q_{2,1}\}$$
where $Q_{1,1},Q_{1,2}\in\span_\R(\M_{6,2})$ and $Q_{2,1}\in\span_\R(\M_{4,2})$.

In step 3, we use Procedure \ref{proc-H} to compute the SOS representation \eqref{eq-S1}
and \eqref{eq-S2} with $E_{3,4,n}$, $\C_{4,2}$, and $\mathbb{C}$ s the input.

%

In Step {\bf S1}, the new variables are $\M_{4,2}=\{m_i,i=1,\ldots,33\}$.

In Step {\bf S2}, using Gaussian elimination to $\C_{4,2}=\{R^{(2)}_{i,1,2}\,:\,i=1,\ldots,589\}$,
we obtain $\C_{4,2,1}$ and $\C_{4,2,2}$ which contain 266 and 182 constraints, respectively.

In Step {\bf S3}, there exist 182 intrinsic constraints. Thus,  ${\widehat{\C}}_{4,2,1}$ contains 448 constraints.
Using $\R$-Gaussian elimination in $\span_\R({\widehat{\C}}_{4,2,1})$
shows that 417 of these 448 constraints are linearly independent, so $N_3 = 417$.

In Step {\bf S4}, we obtains $\widehat{\mathbb{C}}(4,2)$ which contain 3 log-concave constraints,
so $N_2=3$.

In Step {\bf S5}, eliminating the non-quadratic monomials in $E_{3,4,2}$
using $\C_{4,2,2}$ to obtain a quadratic form in $m_i$ and then simplifying the quadratic form using $\C_{4,2,1}$, we have
{\small
\begin{equation*}
\begin{array}{ll}
\widehat{E}_{3,4,2}\!\!\!\!
&=-{4}m_{2}m_{7}+{3}m_{3}^2+\frac{135}{8}m_{26}^2-{3}m_{3}m_{14}+\frac{3}{4}m_{14}m_{19}-{156}m_{32}^2+{6}m_{21}^2-{36}m_{21}m_{27}\\
&+{12}m_{20}m_{26}+{2}m_{20}^2+\frac{135}{4}m_{18}^2+{96}m_{15}m_{32}-{24}m_{15}m_{24}-{72}m_{15}m_{27}+{8}m_{15}m_{21}+{6}m_{15}^2\\
&-{36}m_{14}m_{23}-{8}m_{13}m_{20}-{2}m_{13}^2-{4}m_{13}m_{14}-{10}m_{12}^2+{8}m_{12}m_{15}-{6}m_{11}^2-{6}m_{10}^2+{16}m_{10}m_{13}\\
&+{12}m_{10}m_{14}+{6}m_{9}^2-{6}m_{8}^2-{8}m_{8}m_{11}+{2}m_{7}^2-{2}m_{6}^2-m_{5}m_{25}+{4}m_{5}m_{20}-{8}m_{5}m_{10}-m_{5}m_{14}\\
&-{6}m_{5}m_{16}+\frac{1}{2}m_{5}^2-{4}m_{4}m_{15}-{4}m_{4}m_{7}+{2}m_{4}^2+\frac{31}{48}m_{14}^2-\frac{249}{8}m_{19}^2
+\frac{5}{4}m_{14}m_{16}+\frac{135}{4}m_{22}^2\\
&-\frac{669}{112}m_{33}^2-\frac{37}{4}m_{25}m_{28}+\frac{31}{48}m_{25}^2+\frac{9}{4}m_{16}m_{28}-\frac{669}{112}m_{29}^2-m_{1}m_{14}
 +\frac{1}{2}m_{1}^2+\frac{39}{2}m_{23}^2-\frac{33}{8}m_{16}^2\\
&+\frac{135}{4}m_{24}^2+{2}m_{2}^2+\frac{135}{8}m_{17}^2+\frac{135}{4}m_{27}^2+\frac{135}{8}m_{28}^2-\frac{327}{10}m_{30}^2-\frac{639}{8}m_{31}^2
+\frac{3}{4}m_{25}m_{26}+\frac{411}{8}m_{26}m_{31}\\
&+\frac{1233}{8}m_{26}m_{33}+\frac{153}{4}m_{14}m_{28}-\frac{63}{4}m_{14}m_{26}+\frac{411}{40}m_{26}m_{29}-\frac{19}{4}m_{16}m_{25}-\frac{63}{4}m_{16}m_{26}-\frac{37}{4}m_{14}m_{17}.
\end{array}
\end{equation*}}

In Step {\bf S6}, using the Matlab program in Appendix A
with $\widehat{E}_{3,4,2}$, $\widehat{\C}_{4,2,1}$ and $\widehat{\mathbb{C}}(4,2)$ as input,
we find an SOS representation for $\widehat{E}_{3,4,2}$. Thus, $C_3(4,2)$ is proved under the log-concave condition.
The Maple  program to prove $C_3(4,2)$ can be found in
https://github.com/cmyuanmmrc/codeforepi/.

\section{Conclusion}
\label{sec-conc}

In this paper, two conjectures concerning the lower bounds
for the derivatives of $H(X_t)$ are considered.
We first consider a conjecture of McKean
$C_2(m,n): (-1)^{m+1}(\d^m/\d^m t)H(X_t)\ge(-1)^{m+1}(\d^m/\d^m t)H(X_{Gt})$
in the multivariate case.
We propose a general procedure to prove inequities similar to $C_2(m,n)$.
Using the procedure, we  prove $C_2(1,n)$ and $C_2(2,n)$.
We notice that $C_2(m,n)$ cannot be proved for $m>2$ and $n>1$ with the procedure
even under the log-concave condition, which motivates us to propose the following weaker conjecture
$C_3(m,n): (-1)^{m+1}(\d^m/\d^m t)H(X_t)\ge(-1)^{m+1}\frac{1}{n}(\d^m/\d^m t)H(X_{Gt})$. 
%
%
%
Using our procedure, we prove   $C_3(3,2),C_3(3,3),C_3(3,4)$ and $C_3(4,2)$ under the log-concave condition.

From $C_2(1,n)$ and $C_2(2,n)$ proved in this paper, the exact lower bounds for
$(-1)^{m+1}(\d^m/\d^m t)H(X_t)$ are  $(-1)^{m+1}(\d^m/\d^m t)H(X_{Gt})$
for $m=1$ and $2$, respectively. The high order cases are widely open and we give a brief summary of the known results below.

First consider the univariate case ($n=1$). 
$C_1(3,1)$ and $C_1(4,1)$ were true~\cite{Cheng2015} 
and $C_1(5,1)$ cannot be proved with the SDP approach\footnote{In this paper, when we say $C_s(m,n)$  cannot be proved with the SDP approach, we mean that the software in Appendix A terminates and gives a negative answer for problem \eqref{eq-tt31}.}~\cite{Zhang2018,GYG2020}.
$C_2(3,1)$, $C_2(4,1)$, and $C_2(5,1)$ were true under the log-concave condition~\cite{Zhang2018}.
$C_2(6,1)$ is considered in this paper. The  software in Appendix A shows 
that $E_{2,6,1}\ge0$ under the log-concave condition. However, due to the accuracy of the SDP solver, we cannot find an explicit SOS representation. So if the SDP software is correct, 
$C_2(6,1)$ is proved under the log-concave condition.
From these results, a reasonable target is to prove $C_1(m,1)$ or $C_2(m,1)$ under the log-concave condition.

For the multivariate case, 
$C_1(3,2)$, $C_1(3,3)$, $C_1(3,4)$ were true
and $C_1(4,2)$ cannot be proved with the SDP approach~\cite{GYG2020}.
In this paper, $C_3(3,2)$, $C_3(3,3)$, $C_3(3,4)$, and $C_3(4,2)$ were proved under the log-concave condition,  and $C_2(3,2)$, $C_2(3,3)$, $C_2(3,4)$, and $C_2(4,2)$ cannot be proved with the SDP approach under the log-concave condition.
From these results, a guess for the lower bound is  $(-1)^{m+1}(\d^m/\d^m t)H(X_t)\ge(-1)^{m+1}A(n)(\d^m/\d^m t)H(X_{Gt})$, where $A(n)$ is a function in $n$ such that $0\le A(n)\le 1$. 

In order to use the SDP approach to prove more difficult problems such as $C_1(3,n)(n>4)$
and $C_3(3,n)(n>4)$ under the log-concave condition, two kinds of improvements
are needed.
First, it is easy to see that the size of $E_s(m,n)$ and the numbers of the constraints
increase exponentially as $m$ and $n$ becomes larger.
Thus, we need to find certain  rules which could be used to simplify the computation.
Second, in many cases, such as $C_2(3,2)$ under the log-concave constraint,
the SDP program terminates and gives a negative answer.
Since the SDP method is not complete for our problem, we do not know whether
an SOS representation exists.
We thus need a complete method to solve problem \eqref{eq-S1}.
Another problem is to find more constraints besides those used in this paper
in order to increase the power of the approach.
%
%

\section*{Acknowledgments}

This work is partially supported by NSFC 11688101 and NKRDP 2018YFA0704705, Beijing Natural Science Foundation (No. Z190004),
and China Postdoctoral Science Foundation (No. 2019TQ0343, 2019M660830).
%

\section*{Appendix A. Sum of square of quadratic forms based on SDP}
\label{app-b}
%
%
%

We first restate the problem.
Let $f=\widehat{E}_{s,m,n}$,
 $g_i=\widehat{R}_i,i=1,\ldots,N_{3}$,
 ${\RC}_j=\sum_{l=1}^{V_j} q_{j,l} h_{j,l},j=1,\ldots,N_2$,
 where $f,g_i,h_{j,s}$ are quadratic forms in $\R[\M_{m,n}]$
 and $q_{j,s}$ are variables to be determined.
For simplicity, let $\textbf{x} = \M_{m,n}=\{x_1,\ldots,x_u\}$, $U=N_3,V=N_2$.
$Q_j = \sum_{l=1}^{V_j} q_{j,l} m_{j,l}$, where $q_{j,l}$ are variables to be found
and $m_{j,s}\in\M_{2d_j,n}$ are monomials in $\R[\P_n]$ of degree $2d_j$
and total order $2d_j$.
We need to compute $p_i,q_{j,l}\in\R$ such that
\begin{eqnarray}
&&f
 -\sum_{i=1}^{U} p_i g_i
 -\sum_{j=1}^{V}\sum_{l=1}^{V_j}q_{j,l}h_{j,l}  =S ,\label{eq-A1}\\
&&Q_{j}= \sum_{l=1}^{V_j} q_{j,l} m_{j,l}\ge0,j=1,\ldots,V\label{eq-A2}
\end{eqnarray}
where $S=\sum_{i=1}^{u} c_i (\sum_{j=i}^{u} e_{ij} x_j)^2$ is an SOS,
$c_i,e_{ij}\in\R$ and $c_i\ge0$.

We first reduce $Q_j$ into quadratic forms.
Let $\textbf{y}_j= \M_{d_j,n} = \{y_{j,1},\ldots,y_{j,w_j}\}$.
Write the monomials of $Q_j$ as quadratic monomials in $y_{j,s}$,
and still denote results as $Q_j$. If $Q_j$ is not a quadratic form
in $y_{j,s}$, then just set the coefficients of those non-quadratic monomials to zero.
Then, constraint \eqref{eq-A2} becomes
\begin{eqnarray}\label{eq-A3}
Q_{j}= \sum_{l=1}^{V_j} q_{j,l} t_{j,l}\ge0,j=1,\ldots,N_2
\end{eqnarray}
where each $t_{j,l}$ is a quadratic monomial in $y_{j,k}$.

A polynomial $f$ in $\R[\textbf{x}]$ is called {\em positive semidefinite}
and is denoted as $f\succeq0$,
if $\forall \tilde{\textbf{x}}\in \R^u, f(\tilde{\textbf{x}})\ge0$.
%

{\noindent\bf Lemma B}.\cite{GYG2020}
Let $f\in \Q[\textbf{x}]$  be a quadratic form. Then $f\succeq0$ if and only if
\begin{equation}
\label{2.1t}
f=\sum_{i=1}^u c_i (\sum_{j=i}^u e_{i,j} x_j)^2,
\end{equation}
where $c_i,e_{i,j}\in\Q$, $c_i\ge0$, and $e_{i,i}\ne0$ if $c_i\ne0$,  for $i=1,\ldots,u$ and $j=i,\ldots,u$.

Based on Lemma B, problem \eqref{eq-A1}  is equivalent to the following problem.
\begin{equation}
\label{eq-A4}
\begin{array}{l}
\exists p_i,\ q_{j,l}\in\R,\, {\rm s.t.}\\[0.2cm]
 f-\sum_{i=1}^{U} p_i g_i -\sum_{j=1}^{V}\sum_{l=1}^{V_j}q_{j,l}h_{j,l}\succeq0,\\
Q_{j}= \sum_{l=1}^{V_j} q_{j,l} t_{j,l}\succeq0,j=1,\ldots,N_2
\end{array}
\end{equation}
Problem (\ref{eq-A4}) can be solved with semidivine programming (SDP).
For details of SDP, please refer to~\cite{Boyd1,Boyd2}.
%
%
A symmetric matrix $\mathcal{M}\in\R^{n\times n}$ is called {\em positive semidefinite}
and is denoted as $\mathcal{M}\succeq0$, if all of its eigenvalues are nonnegative.
Rewrite
$$\begin{array}{ll}
f(\textbf{x})=\textbf{x}C\textbf{x}^T,\ \ g_i(\textbf{x})=\textbf{x}A_i\textbf{x}^T,\ i=1,\ldots,U,\\
h_{j,l}(\textbf{x})=\textbf{x}A_{j,l}\textbf{x}^T,\ \
t_{j,l}=\textbf{y}_jB_{j,l}\textbf{y}_j^T,\ \ j=1,\ldots,V;\ i=1,\ldots,V_j.
\end{array}$$
where $C$, $A_{j,i}$ are $u\times u$ real symmetric matrices, $B_{j,i}$ is $w_{j}\times w_{j}$ real symmetric matrix.
Then, problem \eqref{eq-A4} is equivalent to the following SDP problem:
\begin{equation}
\label{eq-sdp3}
\begin{array}{ll}
\textrm{min}_{p_i,q_{j,l}\in\R}\ \ &0\\[0.2cm]
\textrm{subject\ to\ }\ &C-\sum_{i=1}^{U}p_iA_i
  -\sum_{j=1}^{V}\sum_{l=1}^{V_j}q_{j,l}A_{j,l} \succeq0,\\[0.2cm]
&\sum_{l=1}^{V_j}q_{j,l} B_{j,l}\succeq0,j=1,\ldots,V
\end{array}\end{equation}
The dual of problem (\ref{eq-sdp3}) is
\begin{equation}\begin{array}{ll}
\label{eq-sdp4}
\textrm{min}_{X\in\R^{u\times u},X_{j,l}\in\R^{n_{j,l}\times n_{j,l}}}\ \ &\langle X,C\rangle\\
\textrm{subject\ to\ }&\langle X,A_i\rangle=0,\ i=1,2,\ldots,U\\
&\langle X,A_{j,l}\rangle-\langle X_{j,l},B_{j,l}\rangle=0,\ j=1,\ldots,V, s=1,\ldots,V_j.
\end{array}
\end{equation}
where $X_1$ and $X_{j,l}$ are symmetric matrices and $\langle\cdot\rangle$ is the inner product by treating matrices as vectors.

Problem (\ref{eq-sdp4}) can be solved with the following Matlab program
which computes $P=(p_1,\ldots,p_{U},
 q_{1,1},\ldots$, $q_{V,V_V})$ with $C$, $A_{j,l}$ and $B_{j,l}$ as the input.
This program uses the CVX package in Matlab~\cite{grant2008} to solve SDPs.
\lstset{language=Matlab}
\begin{lstlisting}
cvx_begin
    variable X(u,u) Xjs(w_j,w_j),  j=1,...,V, s=1,...,V_j symmetric
    dual variable P
    minimize(trace(C*X))
    subject to
    [trace(A_1*X),trace(A_i*X), i=1,...,U,
     trace(A_js*X-B_js*Xjs), j=1,...,V, s=1,...,V_j,
     zeros(r,1):P;
    X == semidefinite(n);
    X_{js} == semidefinite(n_js); j=1,...,V, s=1,...,V_j
cvx_end
\end{lstlisting}

\end{document}